\documentclass[journal,twocolumn,10pt]{IEEEtran} 

\usepackage[section]{placeins}
\usepackage{graphics} 
\usepackage{epsfig} 
\usepackage{mathptmx} 
\usepackage{amsthm}
\usepackage{amsmath} 
\usepackage{amssymb}  
\usepackage{cuted}
\usepackage{float}
\usepackage{color}
\usepackage{url}
\usepackage{bm}
\usepackage{dblfloatfix}

\usepackage[ruled,vlined,linesnumbered,noresetcount]{algorithm2e} 
\usepackage{etoolbox}

\usepackage [left=0.9in,right=0.9in,top=1in,bottom=1in]{geometry}

\makeatletter
\patchcmd{\@algocf@start}
  {-1.5em}
  {0pt}
  {}{}
\makeatother

\usepackage{tabularx}

\usepackage{caption,subcaption,cite}

\newtheorem{theorem}{Theorem}
\newtheorem{corollary}{Corollary}
\newtheorem{lemma}{Lemma}
\newtheorem{proposition}{Proposition}
\newtheorem{definition}{Definition}
\newtheorem{remark}{Remark}
\newtheorem{example}{Example}

\makeatletter
\def\ps@IEEEtitlepagestyle{%
  \def\@oddfoot{\mycopyrightnotice}%
  \def\@evenfoot{}%
}
\def\mycopyrightnotice{%
  \begin{minipage}{\textwidth}
  \centering \scriptsize
  Copyright~\copyright~2021 IEEE. Personal use of this material is permitted. Permission from IEEE must be obtained for all other uses, in any current or future media, including reprinting/republishing this material for advertising or promotional purposes,creating new collective works, for resale or redistribution to servers or lists, or reuse of any copyrighted component of this work in other works. 
  \end{minipage}
}
\makeatother
\begin{document}
%
\title{Duplicity Games for Deception Design with an Application to Insider Threat Mitigation}
%
%
%

\author{Linan~Huang,~\IEEEmembership{Student Member,~IEEE,}
        and~Quanyan~Zhu,~\IEEEmembership{Member,~IEEE}
\thanks{This  paper  has  been  accepted  for  publication  in  IEEE Transactions on Information Forensics and Security}
\thanks{This work is partially supported by grants SES-1541164, ECCS-1847056, CNS-2027884, and BCS-2122060 from National Science Foundation (NSF), DOE-NE grant 20-19829 and grant W911NF-19-1-0041 from Army Research Office (ARO).
}
\thanks{L. Huang and Q. Zhu are with the Department
of Electrical and Computer Engineering, New York University,
Brooklyn, NY, 11201, USA. E-mail:\{lh2328,qz494\}@nyu.edu}
\thanks{Digital Object Identifier 10.1109/TIFS.2021.3118886}
}

\maketitle
\begin{abstract}
Recent incidents such as the Colonial Pipeline ransomware attack and the SolarWinds hack have shown that traditional defense techniques are becoming insufficient to deter adversaries of growing sophistication. 
Proactive and deceptive defenses are an emerging class of methods to defend against zero-day and advanced attacks. This work develops a new game-theoretic framework called the duplicity game to design deception mechanisms that consist of a generator, an incentive modulator, and a trust manipulator, referred to as the GMM mechanism. 
We formulate a mathematical programming problem to compute the optimal GMM mechanism, quantify the upper limit of enforceable security policies, and characterize conditions on user's identifiability and manageability for cyber attribution and user management. 
We develop a separation principle that decouples the design of the modulator from the GMM mechanism and an equivalence principle that turns the joint design of the generator and the manipulator into the single design of the manipulator. 
A case study of dynamic honeypot configurations is presented to mitigate insider threats. The numerical experiments corroborate the results that the optimal GMM mechanism can elicit desirable actions from both selfish and adversarial insiders and consequently improve the security posture of the insider network. In particular, a proper modulator can reduce the \textcolor{black}{incentive misalignment} between the players and achieve win-win situations for the selfish insider and the defender. 
Meanwhile, we observe that the defender always benefits from faking the percentage of honeypots when the optimal generator is presented. 
\end{abstract}
%
\begin{IEEEkeywords}
Bayesian persuasion, proactive defense,  mechanism design, insider threat, cyber deception, cyber attribution, cyber trust, incentive mechanism 
\end{IEEEkeywords}

%
\IEEEpeerreviewmaketitle

%
%
%
%


\section{Introduction}
\IEEEPARstart{C}{yber} deception is an emerging proactive defense technique against increasingly sophisticated attacks, including Advanced Persistent Threats (APTs), insider threats, and supply chain attacks. Defensive deception technologies, such as Moving Target Defense (MTD) \cite{jajodia2011moving} and honeypots \cite{bringer2012survey}, create uncertainties and misinformation for adversaries to misdirect their perception and decision processes \cite{al2019autonomous}. 
\textcolor{black}{
Among many success stories from industry, it has been shown in  \cite{instance1290} that deception technology has successfully reduced the data breach costs by 51.4\% and per analyst costs by 32\% for the Security Operations Center (SOC). 
}
An important application of cyber deception is to defend systems from insider threats. Harmful behaviors of \textcolor{black}{inadvertent insiders or insiders with malicious intentions} can lead to compromises of sensitive data and disruptions in the organization's normal operations \cite{Mitigation}. 
Deception technologies provide promising proactive solutions to detect unwarranted behaviors and deter the insiders from  wrongdoing, e.g., \cite{spitzner2003honeypots}. 

The design of successful defensive deception relies on a formal approach that quantifies the strategic interactions of the three classes of players, including a defender, users, and adversaries. A useful framework to design cyber deception mechanisms needs to capture three main features.
First, the defender, the users, and the adversaries are strategic players with clear but imperfectly aligned objectives or incentives. Second, the defender cannot distinguish adversaries from the normal users. For example, the defender does not know who is an adversarial insider when designing a security policy for the network. Apart from this, the defender cannot distinguish the type of users in the network concerning their objectives, resources, and trust values. 
Third, a sophisticated adversary behaves stealthily and intelligently, e.g., by conducting successful reconnaissance or acting like a normal user to gain access or trust.  

In this work, we propose \textit{Duplicity Games} (DG) as a mechanism design framework for defensive deception to elicit desirable security outcomes when a defender, normal users, and adversaries interact to attain their individual objectives.
A DG is a two-stage game between a defender and a normal/adversarial user with two-sided asymmetric information. 
The defender, or the defensive deceiver, has private information of the system state. 
The user has a private type, which characterizes the user's objectives, trustworthiness, and attributes, e.g., normal or adversarial. 
At the first stage of the game, the defender designs three composable components of the mechanism, i.e., a \emph{generator}, an \emph{incentive modulator}, and a \emph{trust manipulator}. 
The generator is a mechanism that stochastically generates signals or security policies based on the system's private information and system constraints. 
The modulator reshapes the user's incentive by creating constrained utility transfers between two players. 
The manipulator distorts the user's prior belief over the unknowns. These three components are together referred to as the GMM mechanism. 
After the mechanism is designed and implemented, the user observes the security policies, updates his trust through the Bayesian rule, and then responds to the GMM mechanism by taking an action that serves his objective.  
The optimal design of the GMM mechanisms anticipates the behaviors of different types of users under a given set of security policies and elicits desirable security behaviors.
The GMM mechanisms we introduce here represent a class of multi-dimensional security mechanisms that control the security policies, the (dis)incentives, and the digital footprints (e.g., feature patterns and configurations of honeypots and normal servers). 


We formulate the design problem into a mathematical programming problem, where the anticipated behavioral outcomes of the users follow \textcolor{black}{the Incentive-Compatible (IC) constraint and the Modulation-Feasible (MF) constraint.}
We use concavification techniques as in \cite{aumann1995repeated, kamenica2011bayesian} to provide a graphical analysis and interpretation of  the GMM mechanism. 
We observe that  the user’s \textit{expected posterior utility} can be fully characterized by Piece-Wise Linear and Convex (PWLC) functions. 
This observation leads to a \textcolor{black}{significantly} reduced number of enforceable security policies and enables an efficient implementation of the GMM mechanism. 
Finally, we show a fundamental \textit{separation principle} in which the defender can design the modulator independently, and an  \textit{equivalence principle} where the joint design of the generator and the manipulator is equivalent to the single design of the manipulator.  

For further elaboration, we use the DG framework to study insider threats and design mitigation strategies to deter and prevent misbehavior in corporate networks. 
The corporate network defender can adaptively configure honeypots and normal servers to counter fingerprinting (i.e., generator), modify \textcolor{black}{the complexity} of the authentication process to change user's incentives (i.e., modulator), and misreport the percentage of honeypots to make use of the user's trust (i.e., manipulator). The design of the GMM mechanism  leads to a set of multi-faceted socio-technical solutions for insider threats, which formalizes the management guidelines for insider threats recommended in \cite{Mitigation}.   
From the generator design, 
we propose the concept of the \textit{motive threshold} to assess the average motive of the entire insider population and the concept of the \textit{deterrence threshold} to measure the adequacy of the honeypots. 
From the modulator design, we illustrate how the proper design of the authentication cost can \textcolor{black}{reduce the misalignment between the insiders' and the defender's incentives.} 
From the manipulator design, we find that the manipulation of the insiders' initial beliefs can harm the defender when there are no deceptive generators, but create an advantage when the optimal generator is applied. 

\subsection{Related Work}
\subsubsection{\textcolor{black}{Game Theory for Cyber Deception}}
\textcolor{black}{
Game theory has been widely applied for proactive defense and cyber deception to enhance the security of cyber-physical systems \cite{zhao2020finite,manshaei2013game,huang2017large,pawlick2021game}. 
Games of incomplete information provide a natural paradigm to quantify the uncertainty and misinformation induced by the deception. 
Exemplary game models include signaling games \cite{pawlick2018modeling,9447822}, 
dynamic Bayesian games with finite \cite{HUANG2020101660} or infinite states \cite{huang9494340},}  
(Bayesian) Stackelberg security games  \cite{feng2017signaling,cranford2018learning,xu2016signaling}, and partially observable stochastic games \cite{horak2017manipulating}. 
These incomplete-information games focus on finding signals and behaviors at the equilibrium for a given mechanism and information structure. 
In this work, we further aim to design the mechanism and exploit the information asymmetry, which proactively enhances cyber security. 

\subsubsection{\textcolor{black}{Incentive Mechanisms and Information Design in Cyber Systems}}
\textcolor{black}{There is rich literature on incentive mechanisms designed to enhance security \cite{7539363,7265043}, efficiency \cite{8355788,7180387}, privacy \cite{8355763} of cyber-physical systems. 
They are applied to wide applications, including crowdsourcing \cite{8355788,8355763}, mobile sensing \cite{7265043}, cloud computing \cite{7180387}, cyber insurance \cite{8913631}, and security as service \cite{7539363,chen2017security}.}  
These incentive mechanisms mainly focus on designing the payoff rules and the allocation rules to incentivize participants' behaviors in the designer's favor. 
\textcolor{black}{Besides incentive design, previous works have also investigated information design by disclosing information strategically, which has been applied to wildlife protection \cite{rabinovich2015information}, congestion mitigation \cite{das2017reducing}, and honeypot configuration \cite{HORAK2019101579}.} 
DGs broaden the scope \textcolor{black}{of these two classes of mechanisms} to the joint design of \textcolor{black}{information, incentive, and trust} to achieve desirable equilibrium outcomes. 


\subsubsection{\textcolor{black}{Insider Threat Mitigation and Incentive Design}}
Previous works, e.g.,  \cite{Mitigation,moore2015effective}, have proposed guidelines to establish effective inside threat mitigation programs. 
\textcolor{black}{
Game-theoretic models have been developed to detect insider threats \cite{KANTZAVELOU2010859} and identify the best response strategy \cite{LIU200875}. 
The recent work \cite{9218982} has incorporated organizational culture and the existing defensive mechanisms into the game model. 
These works provide a quantitative understanding of insider threats but overlook the human aspects, such as compliance and incentives, which are fundamental and challenging problems for insider threat mitigation. 
The authors in \cite{casey2015compliance} have used signaling games to model compliant and non-compliant insiders and adopted a feedback loop to control their compliance.} 
This work uses honeypots as a way to detect and monitor the misbehavior of the insiders and aims to formalize the design of such guidelines, e.g., detection, incentives, and penalties. 

\subsection{Notations and Organization of the Paper}
Calligraphic letter $\mathcal{A}$ defines a  set. 
The notation $\Delta \mathcal{A}$ represents the set of probability distribution over $\mathcal{A}$ and $|\mathcal{A}|$ represents its cardinality. 
\textcolor{black}{We summarize main notations for the general model and the case study
in Table \ref{table:notationGeneral} and  Table \ref{table:notationCasestudy}, respectively.} 
The rest of the paper is organized as follows. Section \ref{sec:DG Model} introduces the DG model. We present the mathematical programming and the concavification method in Section \ref{sec:MathProg} and \ref{sec:Concavification}, respectively. 
Section \ref{sec:case study} presents a case study of honeypot configuration to mitigate insider threats and Section \ref{sec:conclusion} concludes the paper.

\begin{table}[t]
\centering
\caption{Summary of notations for DG-GMM. 
\label{table:notationGeneral}}
\textcolor{black}{
\begin{tabularx}{\columnwidth}{l X} 
 \hline
 \textbf{Variable}      & \textbf{Meaning}      \\ \hline
 $N,M,K$   & Number of states, types, and actions.    \\ 
  $b(\cdot)\in \Delta \mathcal{X}$      & True probability distribution of the  state.  \\
 $b_U(\cdot|\theta)\in \Delta \mathcal{X}$      & User's initial belief of the state under $\theta$.  \\
$b_D(\cdot|x)\in \Delta \Theta$      & Defender's initial belief of the type at $x$.    \\ 
$\mathbf{p}^0:=[p^0_1,\cdots,p^0_N]$       & Common prior belief in vector form.  \\
$\mathbf{p}:=[p_1, \cdots, p_N]$     & Common posterior belief in vector form.     \\
$a^*_{\theta}(b_U^{\pi})$ or $a_{\theta}^*(\mathbf{p})$    & Optimal response action of a type-$\theta$ user to maximize his expected posterior utility.  \\
$\bar{v}_D(\pi,\mathbf{p}^0)$ & Defender's \textit{expected posterior utility} under generator $\pi$ and common prior belief $\mathbf{p}^0$.\\ 
$\tilde{v}_D(\mathbf{p}^0)=\bar{v}_D(\pi^0,\mathbf{p}^0)$ & Defender's \textit{prior utility} where generator $\pi^0$ contains zero information. \\
$V_D(\mathbf{p}^0)=\bar{v}_D(\pi^*,\mathbf{p}^0)$ & Defender's \textit{optimal posterior utility} where generator $\pi^*$ is optimal. \\
  $s_{\{a^1,a^2,\cdots,a^M\}}$       & Security policy that requires the user of type $\theta_l\in \Theta$ to take action $a^l\in \mathcal{A}$ for all $l\in \{1,2,\cdots,M\}$. \\
\hline
\end{tabularx}
}
\end{table}

\begin{table}[t]
\centering
\caption{Summary of notations in the case study. 
\label{table:notationCasestudy}}
\textcolor{black}{
\begin{tabularx}{\columnwidth}{l X} 
 \hline
\textbf{Variable} & \textbf{Meaning} \\ \hline
$p_D^{0,H}$       & Defender's prior belief of a node being a honeypot.\\
$p_U^{0,H}$       & Insider's prior belief of a node being a honeypot. \\
$p_U^H$         & Insider's posterior belief of a node being a honeypot. \\
 $q^g$ & Percentage of  selfish insiders. \\ 
  $q^b$ & Percentage of adversarial insiders. \\ 
  $r_U \phi^0$ & Insider’s authentication cost. \\
$t^g(\phi^0)$  & Decision thresholds of the selfish insiders. \\
$t^b(\phi^0)$  & Decision thresholds of the adversarial insiders. \\
\hline
\end{tabularx}
}
\end{table}

\section{Duplicity Game Model}
\label{sec:DG Model}
\textcolor{black}{We present a motivating example of insider threat mitigation in Section \ref{sec:motivatingexample}.} 
Then, we present the structure of DG in Section \ref{sec:gameElements} and the timeline of the GMM mechanism design in Section \ref{sec:timeline}, respectively. 
\textcolor{black}{Finally, we illustrate the relation of the DG-GMM mechanism to the Bayesian persuasion framework in Section \ref{sec:relation2BP}.}


 
\textcolor{black}{ 
\subsection{Motivating Example of Insider Threat Mitigation}
\label{sec:motivatingexample}
Insider threats have been a long-standing problem in cybersecurity. 
Due to their information, privilege, and resource advantages over external attackers, insider threats can circumvent classical defense techniques such as intrusion prevention and detection systems. 
As a result, defensive deception methods, such as honeypots, have been used for insider threat detection and mitigation (see e.g., \cite{spitzner2003honeypots,yamin2019implementation}).}  
Theoretically, honeypots are assumed to achieve a zero false-positive rate and low false-negative rate by generating decoys accessed only by attackers. 
This assumption may not hold for insider threats. 
On the one hand, non-adversarial insiders who are curious or error-prone can access honeypots, \textcolor{black}{which intensifies alert fatigue}. 
On the other hand, adversarial insiders can access the internal information and fingerprint honeypots \cite{dahbul2017enhancing,morishita2019detect} using features such as open ports, protocols, and error responses. 
\textcolor{black}{
To address these two challenges, we need to configure the honeypot and the normal servers strategically. 
The configuration needs to elicit desirable behaviors from both adversarial and non-adversarial insiders even though they have the same insider information. 
This work introduces three configuration methods that can be used independently or jointly; i.e., configure the feature pattern adaptively (see Example \ref{example:featurepattern} for details), prolong or shorten the authentication time to change insiders' incentives, and misreport the percentage of honeypots to make use of the insiders' trust.}

\subsubsection{\textcolor{black}{Categorization of Insiders' Motives}}
\label{sec:Categorization}
An insider's motive can be roughly classified into seven subcategories based on the VERIS Community Database (VCDB) \cite{WinNT3}. 
We divide these subcategories of motives into three classes of motives: selfish, adversarial, and unintentional. They make up $12\%$, $26\%$, and $62\%$, respectively. 
The class of selfish motives includes fun, convenience, fear, or ideology. The adversarial  motives include espionage, financial gain, or grudge. 
The category of unintentional motives refers to the negligent insiders who take no notice of the deceptive configuration and make habitual decisions. 
\textcolor{black}{The incentives of unintentional insiders are often uncontrollable through incentives. Our incentive design mechanism here focuses on the class of the selfish insiders, who seek self-interest, and the adversarial ones, who seek to sabotage the organization.} 

\subsubsection{\textcolor{black}{Corporate Network with Insiders and Honeypots}}
 Fig. \ref{fig:scenarioDiag} illustrates a corporate network with honeypots (denoted by $x^H$) and normal servers (denoted by $x^N$) as nodes. 
  \begin{figure}[h]
\centering
\includegraphics[width=0.95\columnwidth]{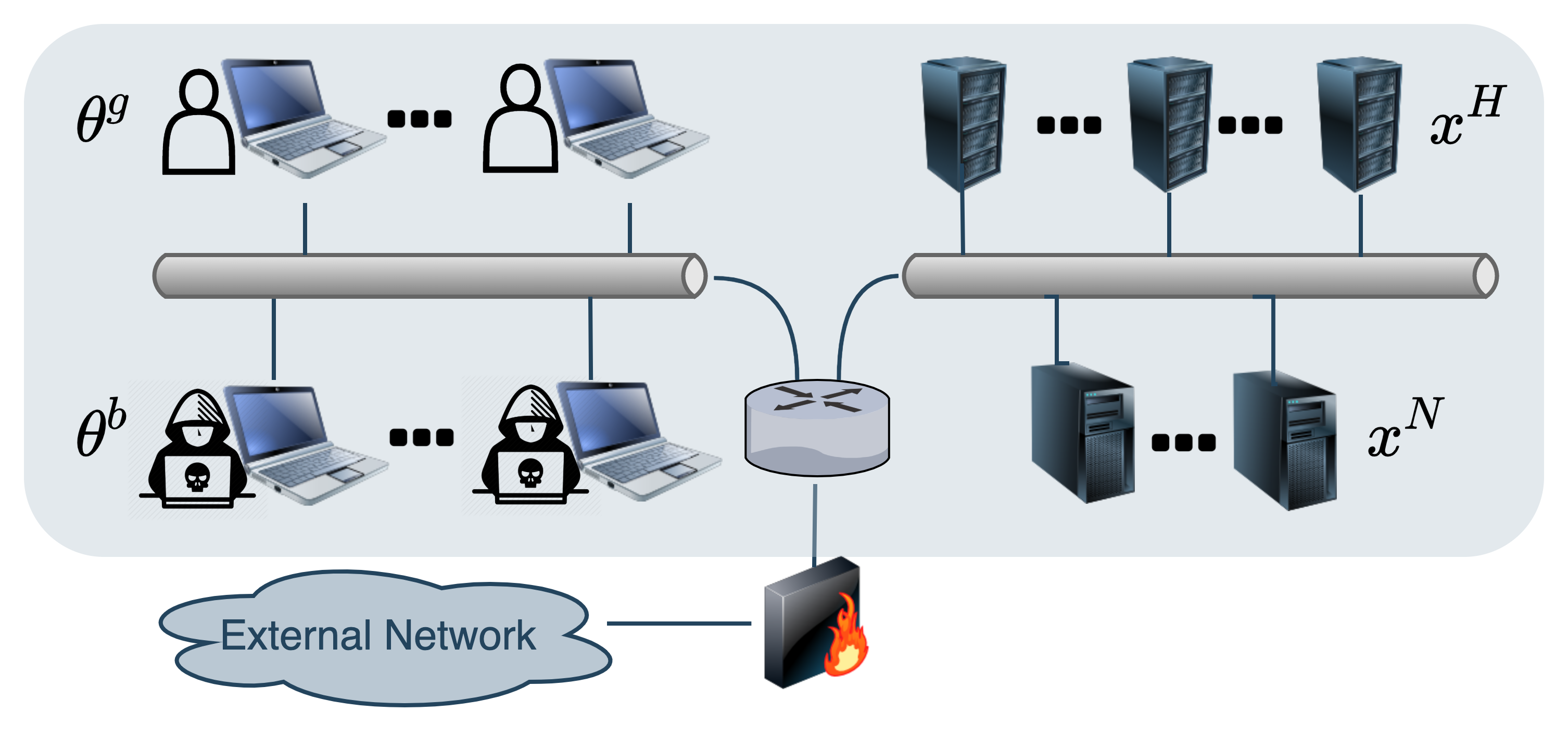}
\caption{ 
An example corporate network consists of normal servers and  honeypots. The light blue background shows the \textcolor{black}{region} of the internal network.  
}
\label{fig:scenarioDiag}
\end{figure}  
\textcolor{black}{
The SOC, or the defender, can privately determine the percentage, the location, and the configuration of honeypots in the corporate network. 
The goal of the defender is to elicit desirable behaviors from the selfish insiders (denoted by $\theta^g$) and the adversarial insiders (denoted by $\theta^b$).} 
Both types of insiders can take harmful actions intentionally yet for different reasons or motives. For example, selfish insiders may violate security rules and abuse their privileges to save time and effort in finishing their tasks. They do not seek to sabotage the organization as the adversarial ones do. 
For each node in the corporate network, an insider can either access it (denoted by action $a_{AC}$) or not (denoted by action $a_{DO}$).

  \begin{figure*}[th]
\centering
\includegraphics[width=0.9\textwidth]{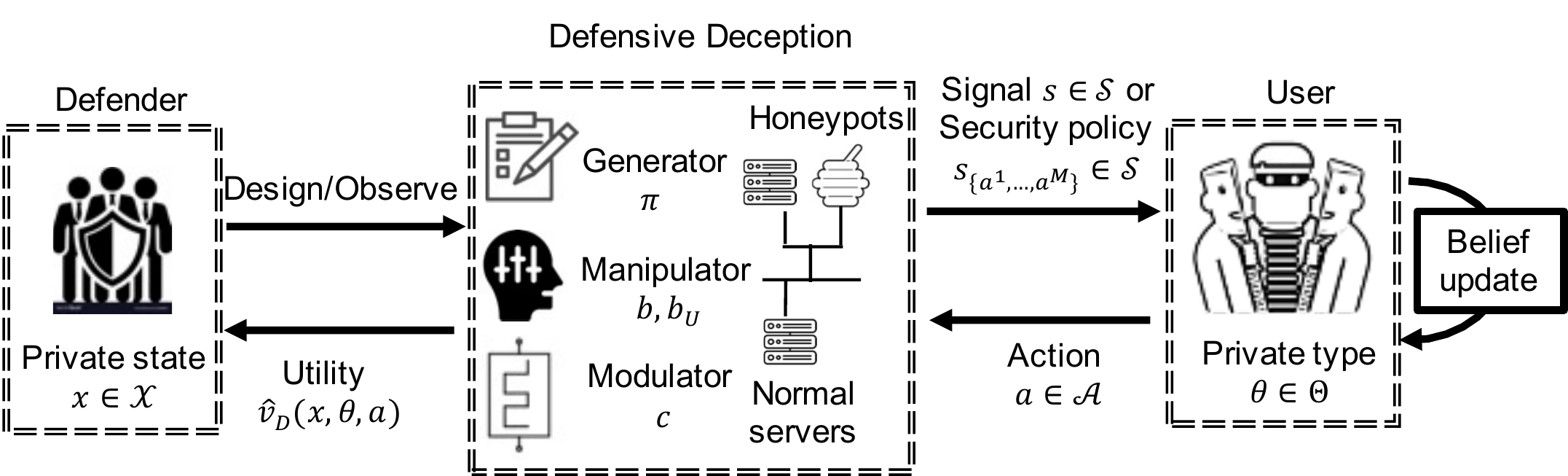}
\caption{
Timeline for the GMM mechanism design. 
}
\label{fig:honeycase}
\end{figure*}  
\subsection{Game Elements}
\label{sec:gameElements}
The DG consists of four elements; i.e., the \textit{basic game} $(\mathcal{X}, \Theta,\mathcal{A}, v_D,v_U, b\in \Delta \mathcal{X})$, the \textit{belief statistics} $(b_D(\cdot|x)\in \Delta \Theta,b_U(\cdot|\theta)\in \Delta \mathcal{X})$, the \textit{information structure} $(\mathcal{S},\pi\in \Pi)$, and the \textit{utility transfer} $(\gamma,c\in \mathcal{C})$. 

\subsubsection{Basic Game}
\label{sec:basic games}
The DG consists of two players $i\in \{D,U\}$, a defender $i=D$ (hereafter she) and a user $i=U$ (hereafter he). 
Define the finite sets of $N$ states, $M$ types, and $K$ actions as $\mathcal{X}:=\{x_1, \cdots,x_N \}$, $\Theta:=\{\theta_1, \cdots,\theta_M \}$, and $\mathcal{A}:=\{a_{DO},a_1, \cdots,a_{K-1}\}$, respectively. 
Action $a_{DO} \in \mathcal{A}$ is the drop-out action. It indicates that the user chooses not to participate in the game and takes no action. 

The game has two-sided asymmetric information. 
The defender can privately observe or know the realization of the state $x\in \mathcal{X}$ from a probability distribution $b\in \Delta \mathcal{X}$.  
For example, in \textcolor{black}{the corporate network in Fig. \ref{fig:scenarioDiag}, $b(x^H)$ and $b(x^N)$ represent the percentages of honeypots and normal servers, respectively.} 
The user does not know each node's state, i.e., whether a honeypot or a normal server. 
The user has a private type $\theta \in \Theta$ that represents his motive, capacity, rationality, or risk perception. 
The user's behaviors are abstracted as an action $a\in \mathcal{A}$. 
The defender can observe the user's action by monitoring and logging but she cannot observe the user's type; e.g., whether the user accesses the confidential data by accident (i.e., the unintentional type), out of self-interest (i.e., the selfish type), or for adversarial purposes (i.e., the adversarial type).  
 The utility functions of the defender and the user, denoted by 
 $v_i: \mathcal{X}\times \Theta \times\mathcal{A} \mapsto \mathbb{R}, i\in \{D,U\}$, depend on the state, type, and action. 
 
\subsubsection{Belief Statistics}
\label{sec:belief statistics}
The user's initial belief of the state under type $\theta\in \Theta$ is $b_U(\cdot|\theta)\in \Delta \mathcal{X}$. 
Since the user does not know the true state distribution $b(\cdot)$, his perceived state distribution $b_U$ can be different from the true one. 
The defender's belief of the user's type at state $x\in \mathcal{X}$ is $b_D(\cdot|x)\in \Delta \Theta$. 
In the game, the defender can design $b$ and $b_U$ through a virtual \textit{trust manipulator}.  
For example, the defender can determine the percentage of honeypots to be 
\textcolor{black}{$b(x^H)$ but report the percentage as $b_U(x^H|\theta)$ to the type-$\theta$} users who determine the percentage of honeypots based on the report without additional information. 
The trust manipulator is \textit{overt} if the user's perceived state distribution equals the true one for all types, i.e., $b_U(x|\theta)=b(x), \forall x\in \mathcal{X}, \forall \theta\in \Theta$. 
Otherwise, the trust manipulator is said to be \textit{covert} as the defender  stealthily manipulates users' initial beliefs. 

\subsubsection{Information Structure}
\label{sec:Information Structure}
 The \textit{information structure} consists of a finite set of signals $\mathcal{S}$ and a generator $\pi\in \Pi: \mathcal{X}\mapsto \Delta  \mathcal{S}$. 
 With a slight abuse of notation, we use $\pi(s|x)$ to represent the probability of signal $s\in \mathcal{S}$ at state $x\in \mathcal{X}$. 
\textcolor{black}{In Example \ref{example:featurepattern} below}, the signal can be interpreted as the \textit{feature patterns},  including protocols, ports, the response time, and the error response. 

\begin{example}[\textcolor{black}{\textbf{Dynamic Feature Pattern Configurations}}]
\label{example:featurepattern}
To defend against honeypot fingerprinting, dynamic \cite{shi2019dynamic} 
and adaptive \cite{wagener2011adaptive,huang2019adaptive} 
configurations have been adopted in honeypots. 
The SOC can also configure normal servers and disguise them as honeypots by generating honeypot-related features \cite{rowe2007defending}. 

Suppose that there are $J$ features that both honeypots and normal servers can generate.  
Denote the value of feature $j\in \{1,\cdots,J\}$ by  $e^j\in \mathcal{E}^j$, where $\mathcal{E}^j$ is a finite set. 
For example, 
the error response feature can take a binary value $e^j\in \mathcal{E}^j=\{0,1\}$ based on whether an abnormal error message appears under intentionally erroneous requests \cite{morishita2019detect}. 
We refer to the tuple of $J$ features as the \textit{feature pattern} denoted by  $s=(e^1,\cdots,e^{J})\in \mathcal{S}:=\prod_{j=1}^{J} \mathcal{E}^j$. 
Then, the feature pattern of each node changes dynamically accordingly to the generator $\pi\in \Pi$; i.e., a honeypot and a normal server generate feature pattern $s\in \mathcal{S}$ with frequency $\pi(s|x^H)$ and $\pi(s|x^N)$, respectively. 
\textcolor{black}{Insiders can use these feature patterns as the digital footprint to fingerprint a node's state, either a honeypot or a normal server.} 
The DG still applies to the case when the SOC cannot configure normal server. In that case, the decision variable $\pi(\cdot|x^N)$ will be taken as fixed.
\end{example}

\subsubsection{Utility Transfer}
\label{sec:utility transfer}
The \textit{utility transfer} consists of a scaling factor $\gamma\in [0,\infty)$ and an incentive modulator $c\in \mathcal{C}: \mathcal{A}\mapsto \mathbb{R}$ which modifies the utilities of the defender and the user to be $\hat{v}_D(x,\theta,a)=v_D(x,\theta,a)+\gamma c(a)$ and $\hat{v}_U(x,\theta,a)=v_U(x,\theta,a)- c(a)$, respectively, for all $x\in \mathcal{X}, \theta\in \Theta, a\in \mathcal{A}$. 
Besides monetary (dis)incentives, $c(a)$ can also represent the additional cost or benefit of taking action $a\in\mathcal{A}$. 
 For example, it captures the authentication time to access a normal server or a honeypot. 
 The defender can determine the authentication time to incentivize the user (i.e., $c(a)<0$) or disincentivize him (i.e., $c(a)>0$) to take the action $a\in \mathcal{A}$. 
Although the modulator $c$ is type-independent, its influence on users is type-dependent. 
 For example, a curiosity-driven insider may lose interest and give up accessing confidential data under a long authentication delay or a convoluted \textit{multi-factor authentication} process. 
 However, an adversarial insider can be persistent if the data access leads to a comparably high financial return. 
  Definition \ref{def:actionDominate} defines a special utility structure where one action $a_k\in \mathcal{A}$ yields the highest benefit for the user of type $\theta\in \Theta$ regardless of the state values. 
  For a user with a dominant action, a generator does not influence the user's belief and action. 
\begin{definition}
\label{def:actionDominate}
An action $a_k\in \mathcal{A}$ dominates (resp. is dominated) under type $\theta\in \Theta$ if $\hat{v}_U(x,\theta,a_k) \geq (\text{resp. } \leq)  \hat{v}_U(x,\theta,a), \forall a\in \mathcal{A}, \forall x\in \mathcal{X}$. 
\end{definition}

 \subsection{Timeline for the GMM Mechanism Design} 
 \label{sec:timeline}
As shown in Fig. \ref{fig:honeycase}, the GMM mechanism design in DGs has two stages to achieve the intended outcomes of the defensive deception. 
 At stage one, the defender designs (resp. observes) the generator $\pi\in\Pi$, the manipulator  $b\in \Delta \mathcal{X},b_U(\cdot|\theta)\in \Delta \mathcal{X}, \forall \theta\in \Theta$, and the modulator $c\in \mathcal{C}$ if these components can (resp. cannot) be designed. 
\textcolor{black}{Based on the realized state value $x$, the generator generates a signal $s\in \mathcal{S}$ with probability $\pi(s|x)$. In the insider threat example, the defender configures the feature pattern $s$ with probability $\pi(s|x^H)$ (resp. $\pi(s|x^N)$) when the node is a honeypot (resp. normal server).}  
 At stage two, the user of type $\theta\in \Theta$ receives the signal $s\in \mathcal{S}$ and obtains his posterior belief $b_U^{\pi}$ of the state using the Bayesian rule, i.e., 
\begin{equation}
\label{eq:BayesUpdate}
b_U^{\pi}(x |\theta,s):=\frac{b_U(x |\theta)\pi(s|x)}{\sum_{x'\in \mathcal{X}} b_U(x' |\theta)\pi(s|x')}, \forall x\in\mathcal{X}. 
\end{equation}
Then, the user of type $\theta\in \Theta$ takes a best-response action denoted by  $a^*_{\theta}(b_U^{\pi})\in\mathcal{A}$ to maximize his expected posterior utility under the posterior belief $b_U^{\pi}$, i.e., 
\begin{equation}
\label{eq:optimalAction}
a^*_{\theta}(b_U^{\pi}) \in \textrm{arg}\max_{a\in \mathcal{A}} \mathbb{E}_{x\sim b^{\pi}_U(\cdot|\theta,s)} [ \hat{v}_U (x,\theta,a)]. 
\end{equation}

The utility of the users is a way to capture the user behavior $a^*_{\theta}$. For example, $a^*_\theta$ can represent how an insider routinely follows the security rules or abuses his privilege for personal gain. 
The defender's goal is to determine the optimal GMM mechanism to proactively prevent undesirable user behaviors and improve the security posture.  
This objective is achieved by maximizing her \textit{expected posterior utility} $\bar{v}_D$ that captures the outcomes of the user's behaviors, i.e., 
$
\bar{v}_D(\pi,b,b_U,c):= \mathbb{E}_{x\sim b(\cdot)} \allowbreak
\mathbb{E}_{s\sim \pi(\cdot | x)}\allowbreak
\mathbb{E}_{\theta\sim b_D(\cdot|x)} 
[ \hat{v}_D(x,\theta,a_{\theta}^*(b_U^{\pi}) )].$ 
Different generators provide the user with different amounts of information about the state.  
Two extreme cases
are defined in Definition \ref{def:zero-information experiment}. 
A signal from a zero-information generator denoted by $\pi^0\in \Pi$ does not change the user's belief, i.e., $b_U^{\pi^0}(x|\theta,s)=b_U(x|\theta), \forall s\in \mathcal{S}, \allowbreak
\forall x\in \mathcal{X}, \forall \theta\in \Theta$. 
Meanwhile, a signal from a full-information generator deterministically reveals the state to the user. 
\begin{definition}[\textbf{Zero- and Full-Information Generators}] 
\label{def:zero-information experiment}
 A generator $\pi\in \Pi$ contains zero information if $\pi(s|x)=\pi(s|x'), \forall s\in \mathcal{S}, \forall x,x'\in \mathcal{X}$. It contains full information if the mapping $\pi: \mathcal{X} \mapsto \mathcal{S}$ is injective.
\end{definition}

Readers can refer to Section \ref{sec:case study} for a case study of insider threat that illustrates the two-stage GMM design. 


 
 \subsection{\textcolor{black}{Relation to Bayesian Persuasion}} 
 \label{sec:relation2BP}
 \textcolor{black}{DG-GMM mechanism design can be viewed as a generalized class of the Bayesian persuasion framework \cite{kamenica2011bayesian} with heterogeneous receivers, two-sided asymmetric information, and a joint design of information, incentive, and trust. If the user's type set $\Theta$ is a singleton and the defender cannot design the modulator and the manipulator, then DG-GMM degenerates to the Bayesian persuasion framework. 
The consolidation of the modulator and the manipulator into the mechanism gives the defender a higher degree of freedom to improve the performance in the deception design. It yet increases the computation complexity as illustrated in Section \ref{sec:MathProg} and causes the violation of Bayesian plausibility in Section \ref{sec:Bayesian plausibility}. 
 }
 
\subsubsection{Violation of Bayesian Plausibility}
\label{sec:Bayesian plausibility}
\textcolor{black}{
The concept of  Bayesian plausibility has been defined in  \cite{kamenica2011bayesian}, which states that the expected posterior belief should equal the prior belief for all $\pi\in \Pi$. 
However, we show in Lemma \ref{lemma:Bayesian plausibility} that the trust manipulator can violate Bayesian plausibility when the user of type $\theta\in \Theta$ holds a different initial belief as the defender, i.e.,  $\exists x\in \mathcal{X}: b(x)\neq b_U(x|\theta)$. 
\begin{lemma}[\textbf{Bayesian Plausibility}]
\label{lemma:Bayesian plausibility}
For all $\pi\in \Pi$ and  $\theta\in\Theta$, the user's expected posterior probability $b^e_U(x|\theta):=\sum_{s\in\mathcal{S}} \sum_{x'\in \mathcal{X}}b(x')\pi(s|x') b_U^{\pi}(x|\theta,s)$  is always a valid probability measure yet  
is Bayesian plausible if and only if 
the defender and the user have the same initial belief $b(x)=b_U(x|\theta), \forall x\in \mathcal{X}$. 
\end{lemma}
\begin{proof}
A generator $\pi\in \Pi$ generates $s$ with probability $\sum_{x'\in \mathcal{X}}b(x')\pi(s|x')$. 
After receiving $s$, the user of type $\theta$ obtains his posterior belief $b_U^{\pi}(x|\theta,s)$ according to \eqref{eq:BayesUpdate}. 
Thus, the expected posterior probability $\sum_{s\in\mathcal{S}} \sum_{x'\in \mathcal{X}}b(x')\pi(s|x') b_U^{\pi}(x|\theta,s)$ is a valid probability measure over $x$. 
The Bayesian plausibility requires $b^e_U(x|\theta)=\sum_{s\in\mathcal{S}}  \frac{\sum_{x'\in \mathcal{X}}b(x')\pi(s|x') }{\sum_{x'\in \mathcal{X}} b_U(x' |\theta)\pi(s|x')} \pi(s|x) b_U(x |\theta)= b_U(x |\theta), \forall x\in \mathcal{X}$, under all $\pi\in\Pi$, which is equivalent to  the condition $b(x)=b_U(x|\theta), \forall x\in \mathcal{X}$. 
\end{proof}
 }
 
\section{GMM Designs by Mathematical Programming} 
\label{sec:MathProg}
In Section \ref{sec:MathProg}, we provide an integrated design of the GMM mechanism by mathematical programming. 
We first elaborate on the relationship between signals and the user's best-response action to introduce the notion of security policies. 
Each signal $s$ from generator $\pi\in \Pi$ updates the user's belief via \eqref{eq:BayesUpdate} and consequently induces the user of type $\theta\in \Theta$ to take the best-response action $a^*_{\theta}(b_U^{\pi})\in \mathcal{A}$.  
Regardless of the signal set $\mathcal{S}$ and the generator $\pi$, these signals can elicit at most $|\mathcal{A}|^{|\Theta|}=K^M$ distinct outcomes; i.e., the user's best-response action $a^*_{\theta}(b_U^{\pi})$ is $a^l$ if his type is $\theta_l$ for all  permutations of $\theta_l\in \Theta, a^l\in \mathcal{A}$. 
We can aggregate signals in $\mathcal{S}$ based on their elicited actions and divide the entire signal set $\mathcal{S}$ into $K^M$ mutually exclusive subsets denoted as  $\mathcal{S}_{\{a^1,a^2,\cdots,a^M \}}, a^l \in \mathcal{A}, l\in \{1,2,...,M\}$. 
Then, the signals in subset $\mathcal{S}_{\{a^1,a^2,\cdots,a^M \}}$ can be interpreted as the \textit{security policy} that requires the user of type $\theta_l$ to take action $a^l$ for all $l\in \{1,2,\cdots,M\}$. 
Without loss of generality,  we use  one aggregated signal $s_{\{a^1,a^2,\cdots,a^M \}} $ to represent the signals in the set $\mathcal{S}_{\{a^1,a^2,\cdots,a^M \}}$. 
Then, the total number of signals are $|\mathcal{S}|=K^M$, and $\pi(\cdot|x)\in \Delta \mathcal{S}$ is a probability distribution over $K^M$ security policies for each state $x\in \mathcal{X}$. 
 The set $\Pi$ naturally contains two feasibility constraints, i.e.,  $ \pi (s_{\{a^1,\cdots,a^{M}\}}|x)\geq 0, \allowbreak
 \forall s_{\{a^1,\cdots,a^{M}\}}\in \mathcal{S}, \forall x\in\mathcal{X}$, and $\sum_{s_{\{a^1,\cdots,a^{M}\}}\in \mathcal{S}  }\pi (s_{\{a^1,\cdots,a^{M}\}}|x)=1, \forall x\in\mathcal{X}$. 
 \textcolor{black}{In Example \ref{exmple:security policies} below, we continue to use the insider threat scenario in Section \ref{sec:motivatingexample} to illustrate how we obtain security policies based on the feature patterns.} 
 
 \begin{example}[\textcolor{black}{\textbf{Security Policies based on Feature Patterns}}]
 \label{exmple:security policies}
 \textcolor{black}{For binary action set $\mathcal{A}=\{a_{DO},a_{AC}\}$ and binary type set $\Theta=\{\theta^g,\theta^b\}$, the feature patterns in Example \ref{example:featurepattern}} can be aggregated into $K^M=4$ categories of security policies. They are $s_{\{a_{DO},a_{DO}\}}$ (i.e., both types of insiders choose $a_{DO}$), $s_{\{a_{DO},a_{AC}\}}$ (i.e., selfish insiders \textcolor{black}{choose} $a_{AC}$ while adversarial insiders choose $a_{DO}$),  $s_{\{a_{AC},a_{DO}\}}$ (i.e., adversarial insiders \textcolor{black}{choose} $a_{AC}$ while selfish insiders choose  $a_{DO}$), and $s_{\{a_{AC},a_{AC}\}}$ (i.e., both types of insiders choose  $a_{AC}$). 
 \end{example}
 
We can rewrite \eqref{eq:optimalAction} concerning security policies as follows, i.e.,  
$\sum_{x\in \mathcal{X}} b_U^{\pi}(x |\theta_l,s_{\{a^1,\cdots,a^{M}\}})  [ \hat{v}_U(x,\theta_l,a^l)-
\hat{v}_U(x,\theta_l,a^h)  ]  \geq 0, \forall s_{\{a^1,\cdots,a^{M}\}}\in \mathcal{S} , \forall a^h \in \mathcal{A}, \forall \theta_l\in \Theta$.
 The defender's expected posterior utility $\bar{v}_D(\pi,b,b_U,c)$ can be equivalently  represented as 
$
\sum_{x\in \mathcal{X}} b(x)  \allowbreak  \sum_{s_{\{a^1,\cdots,a^{M}\}}\in \mathcal{S}  }  \allowbreak 
\pi(s_{\{a^1,\cdots,a^{M}\}}|x)  \allowbreak 
\sum_{\theta_l\in \Theta} 
b_D(\theta_l|x)  
\hat{v}_D(x,\theta_l,a^l)
$. 
Replacing $b_U^{\pi}$ with \eqref{eq:BayesUpdate}, we formulate the GMM mechanism design as the following constrained optimization COP. 
\begin{equation*}
\begin{split}
& \text{(COP):}   \quad 
r:=\sup_{\pi\in \Pi,b,b_U, c\in \mathcal{C}} \quad  \bar{v}_D(\pi,b,b_U,c) \\ 
& (\text{IC}) 
\sum_{x\in \mathcal{X}}  [ \hat{v}_U(x,\theta_l,a^l)
- 
\hat{v}_U(x,\theta_l,a^h)  ] \pi(s_{\{a^1,\cdots,a^{M}\}}|x)   \\
 & \quad\quad 
 b_U(x |\theta_l)
 \geq 0,
\forall s_{\{a^1,\cdots,a^{M}\}}\in \mathcal{S} , \forall a^h \in \mathcal{A}, \forall \theta_l\in \Theta.  \
\\ 
& (\text{MF}) \  c(a_{DO})=0.
\end{split}
\end{equation*} 
\textcolor{black}{
The decision variables $\pi$, $b$, $b_U$, and $c$ are vectors of dimension $N\times K^M$, $N$, $N\times M$, and $K$, respectively. 
The feasibility constraint contained in  $\Pi$ and the Incentive-Compatible (IC) constraint induce $N\times K^M+1$ and $K^M\times K\times M$ constraints, respectively.  
}

Denote $b^*,b_U^*,\pi^*,c^*$ as the maximizers of COP and $r$ as the value of the objective function under the maximizers.  
The (IC) constraint requires all security policies from the generator to be  compatible with the user's incentives; 
i.e., the user receives the maximum benefit on average when taking the action required by the security policy. 
A security policy cannot be generated if it is not incentive-compatible. 
Based on the (IC) constraint, we define the credible and the optimal generators in Definition \ref{def:Feasible and credible Generator} and enforceable security policies in Definition \ref{def:Enforceable Policies}.  
\begin{definition}[\textbf{Credible and Optimal Generators}]
\label{def:Feasible and credible Generator}
A generator $\pi\in \Pi$ is called credible if it satisfies \text{(IC)}. 
A credible generator is called optimal if it maximizes COP. 
\end{definition}
\begin{definition} [\textbf{Enforceable Security Policies}]
\label{def:Enforceable Policies}
For a given generator $\pi\in \Pi$, a security policy $s_{\{a^1,\cdots,a^{M}\}}\in \mathcal{S}$ is enforceable (resp. unenforceable) if $\exists x\in \mathcal{X}$ such that $\pi(s_{\{a^1,\cdots,a^{M}\}}|x) \neq 0$ (resp. 
$\pi(s_{\{a^1,\cdots,a^{M}\}}|x) = 0, \forall x\in \mathcal{X}$).  
\end{definition}

\textcolor{black}{The Modulation-Feasible (MF) constraint results from the fact} 
that the defender cannot modulate the user's incentive if the user does not participate in the game. 
Although the co-domain of $c$ is $\mathbb{R}$, Theorem \ref{thm:feasibleBounded} shows that the optimal utility transfer $c^*\in \mathcal{C}$ has to remain bounded due to the user's potential threat of taking the drop-out action $a_{DO}$. 
We define the following shorthand notations for Theorem \ref{thm:feasibleBounded}, i.e.,  $\underline{c}(\theta,a):=\max_{x\in \mathcal{X}} v_U(x, \theta, a)-v_U(x, \theta, a_{DO})$, $\bar{r}=\max_{x\in \mathcal{X}} \mathbb{E}_{\theta\sim b_D}[\max_{a\in \mathcal{A}} v_D(\theta,x,a)] $ and $\underline{r}=\min_{x\in \mathcal{X}}\mathbb{E}_{\theta\sim b_D}[\min_{a\in \mathcal{A}} v_D(\theta,x,a)] $. 

\begin{theorem}[\textbf{Feasibility and Design Capacity}]
\label{thm:feasibleBounded}
COP is feasible and bounded. 
The upper bound of $r$ is $ \max \{\max_{x\in \mathcal{X}}\mathbb{E}_{\theta\sim b_D} [v_D(x,\theta,a_{DO})],  \allowbreak \bar{r} 
+ \gamma  \max_{ a\in \mathcal{A},  \theta\in \Theta} \underline{c}(\theta,a) \}$ and the lower bound is $\underline{r}$.  
\end{theorem}

\begin{proof}
We first prove the feasibility. 
Define shorthand notation $a^{*,l}:=arg\max_{a\in \mathcal{A}} \mathbb{E}_{x\sim b_U(x|\theta_l)}[v_U (x,\theta_l,a)-c(a)], \forall l\in \{1,\cdots,M\}$, as the optimal action of the user of type  $\theta_l\in \Theta$ under any feasible prior belief $b_U(x|\theta_l)$ and modulator $c\in \mathcal{C}$. 
Then, the zero-information generator $\pi^0(s_{(a^{*,1}, \cdots, a^{*,M})} |x)=1, \forall x\in \mathcal{X}$, is a feasible solution to COP. 

We prove the boundedness in two steps. We first consider $c(a)=0, \forall a\in \mathcal{A}$. 
Since all decision variables $b, \pi, b_D$ are probability measures, we obtain the upper bound $\bar{r}$ and the low bound $\underline{r}$ of $r$.   
In the second step, we turn the modulator $c$ into a free decision variable with the (MF) constraint. 
Since $c(a)=0, \forall a\in \mathcal{A}$, is a feasible solution, the maximum value of COP does not increase. 
Thus, the value of $\underline{r}$ is bounded. 
To show that the value of $\bar{r}$ is bounded in step two, we focus on 
action $a_j\in \mathcal{A}$, if it exists, that results in a non-negative maximizer  $c^*(a_j)$. 
On the one hand, if $\underline{c}(\theta,a_j)\leq 0, \forall \theta\in \Theta$, then the drop-out action $a_{DO}$ dominates for all types and $r=\max_{b\in \Delta \mathcal{X}} \mathbb{E}_{x\sim b}\mathbb{E}_{\theta\sim b_D} [v_D(x,\theta,a_{DO})] \leq \max_{x\in \mathcal{X}}\mathbb{E}_{\theta\sim b_D} [v_D(x,\theta,a_{DO})]$.
On the other hand, if there exists a type $\theta\in \Theta$ where $\underline{c}(\theta,a_j)>0$ and 
$ c^*(a_j) \geq \underline{c}(\theta,a_j)$, then the user of type $\theta$ will choose the drop-out action $a_{DO}$. 
Thus, 
$r \leq  \gamma  \max_{ a\in \mathcal{A},  \theta\in \Theta} \underline{c}(\theta,a)$. 
\end{proof}

The upper and lower bounds provide the design capacity of the GMM mechanism. 
COP is unbounded without the (MF) constraint as the defender can arbitrarily increase (resp. decrease) the value of $r$ by letting $c(a)$ be an arbitrarily large (resp. small) constant. 
If $c(a)=0, \forall a\in \mathcal{A}$, we can transform COP into a Linear Program (LP) by introducing the following variables, i.e.,  $\eta(s_{\{a^1,\cdots,a^{M}\}},x):=b(x)\pi (s_{\{a^1,\cdots,a^{M}\}}|x)$ and  $\eta_U(\theta, s_{\{a^1,\cdots,a^{M}\}},x):=b_U(x|\theta)\pi (s_{\{a^1,\cdots,a^{M}\}}|x) $.  These new variables take non-negative values and satisfy the following  constraints, i.e., $\sum_{x\in \mathcal{X}, s_{\{a^1,\cdots,a^{M}\}}\in \mathcal{S} }\eta=1$ and $\sum_{x\in \mathcal{X}, s_{\{a^1,\cdots,a^{M}\}}\in \mathcal{S} } \eta_U=1, \forall 
\theta\in \Theta$. 
After we have solved the LP, 
we can obtain the initial beliefs by $b(x)=\sum_{s_{\{a^1,\cdots,a^{M}\}}\in \mathcal{S}} \eta(s_{\{a^1,\cdots,a^{M}\}},x)$ and  $b_U(x|\theta)=\sum_{s_{\{a^1,\cdots,a^{M}\}}\in \mathcal{S}} \eta_U(\theta, s_{\{a^1,\cdots,a^{M}\}},x)$ for all $x\in \mathcal{X}, \theta\in \Theta$.  



\section{Graphical Analysis of GMM Designs} 
\label{sec:Concavification}
In Section \ref{sec:MathProg}, we aggregate signals into $K^M$ equivalent security policies to relate them with the user's best-response action. 
In Section \ref{sec:Concavification}, we directly analyze the posterior belief and the action as each signal uniquely determines a posterior belief. 
Throughout Section \ref{sec:Concavification}, we focus on the \textit{overt} trust manipulator defined in Section \ref{sec:belief statistics}, i.e., $b_U(x|\theta)=b(x), \forall x\in \mathcal{X}, \theta\in \Theta$. 
Define $p^0_j:=b(x_j), \forall j\in \{1,\cdots,N\}$, and the common prior belief in the vector form as $\mathbf{p}^0:=[p^0_1,\cdots,p^0_N]$.  
Since different types of users have the same initial beliefs, the posterior beliefs are also the same. 
Denote $p_j\in[0,1]$ as the user's posterior belief under state $x_j\in \mathcal{X}, \forall j\in \{1, \cdots, N\}$. 
Define the belief vector $\mathbf{p}:=[p_1, \cdots, p_N]$ and the
utility vector $\mathbf{\hat{v}}_U(\theta, a):= [ \hat{v}_U (x_1,\theta,a), \cdots, \hat{v}_U (x_N,\theta,a)]'$ where notation $'$ denotes the matrix transpose. 
For both the prior and the posterior belief vectors, the total probability is one, i.e., $\sum_{n=1}^N {p}^0_n=1$ and $\sum_{n=1}^N {p}_n=1$. 

Section \ref{subsection:benchmark} provides the optimal generator design under the benchmark case where the defender can neither modify the user's incentive, i.e., $c(a)=0,\forall a\in \mathcal{A}$, nor manipulate their initial beliefs.  
Section \ref{subsec:Modulator}  incorporates the modulator and the manipulator into the GMM mechanism design. 

\subsection{Generator Design under the Benchmark Case}
\label{subsection:benchmark}
We rewrite \eqref{eq:optimalAction} in its matrix form as $a_{\theta}^*(\mathbf{p})\in arg\max_{a\in \mathcal{A}} \allowbreak
\mathbf{p} \mathbf{\hat{v}}_U(\theta, a)$. 
Since $\mathbf{p} \mathbf{\hat{v}}_U(\theta, a)$ is an affine function of $\mathbf{p}$ for any action $a\in \mathcal{A}$, maximizing $\mathbf{p} \mathbf{\hat{v}}_U(\theta, a)$ over $a$ in the convex domain $\mathbf{p}\in \Delta \mathcal{X}$ results in a Piece-Wise Linear and Convex (PWLC) function as summarized in Proposition \ref{proposition:PWLC}. The proof of convexity follows directly from the fact that $a_{\theta}^*(\mathbf{p})$ is the point-wise maximum of a group of affine functions over $\mathbf{p}$.  
\begin{proposition}
\label{proposition:PWLC}
The user's expected posterior utility under a give type $\theta\in \Theta$, i.e., $\max_{a\in \mathcal{A}} \mathbf{p} \mathbf{\hat{v}}_U(\theta, a) $, is continuously PWLC with respect to vector $\mathbf{p}\in \Delta\mathcal{X}$. 
\end{proposition}

We visualize $\max_{a\in \mathcal{A}} \mathbf{p} \mathbf{\hat{v}}_U(\theta, a)$ under a binary state set in Fig. \ref{fig:binaryPWLC}. 
When $N=2$, we can use the first element $p_1$ as the $x$-axis to uniquely represent the posterior belief $\mathbf{p}\in \Delta\mathcal{X}$. 
The four belief thresholds, i.e., $0, t_1^{\theta},t_2^{\theta}$, and $1$, divide the entire belief region of $p_1\in [0,1]$ into three sub-regions.  
The user of type $\theta$ takes action $a_{K-1}$ if his posterior belief belongs to the sub-region $p_1\in [0,t_1^{\theta}]$, action $a_1$ if $p_1\in [t_1^{\theta},t_2^{\theta}]$, and action $a_{DO}$ if $p_1\in [t_2^{\theta},1]$. 
Although action $a_2$ is not dominated under type $\theta$ based on Definition \ref{def:actionDominate}, it is inactive over $p_1\in [0,1]$. 
\begin{figure}[h]
\centering
\includegraphics[width=1\columnwidth , height=0.4\columnwidth
]{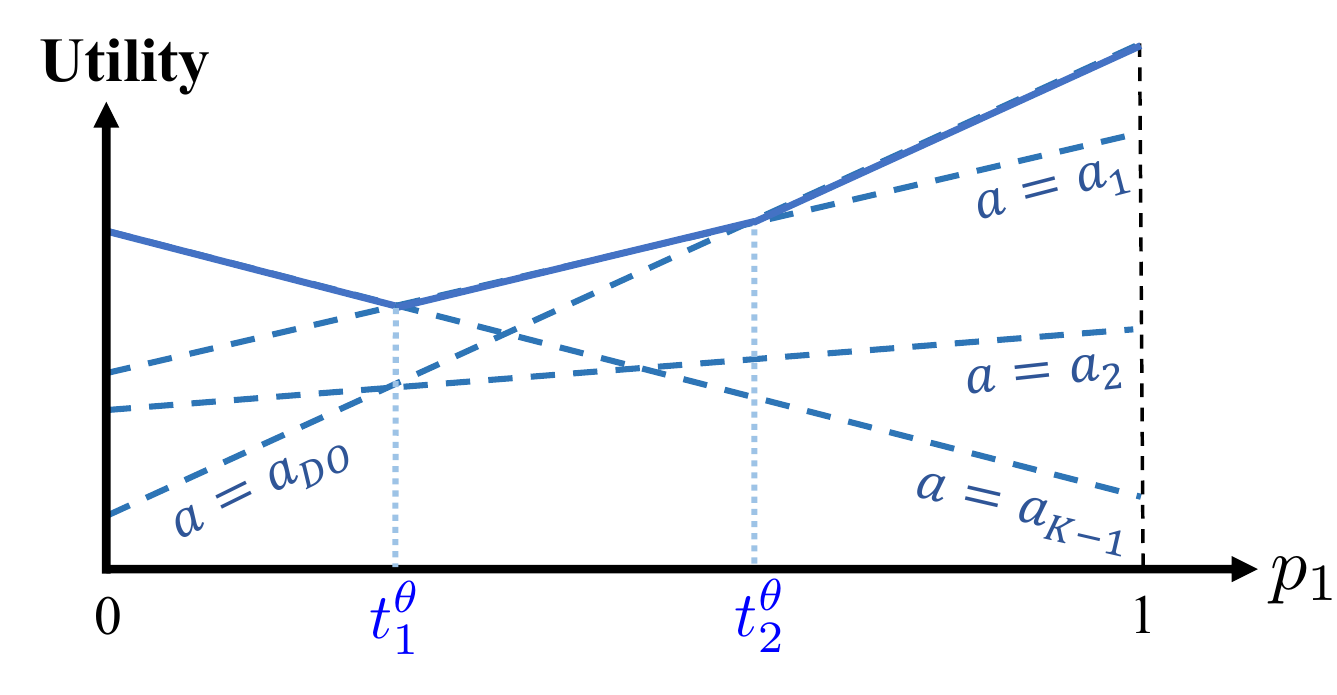} 
\caption{ \label{fig:binaryPWLC}
The expected posterior utility of the user of type $\theta \in \Theta$ versus posterior belief $p_1\in [0,1]$. 
The solid lines represent the utility $\max_{a\in \mathcal{A}} \sum_{n=1}^{N} p_n  \hat{v}_U (x_n,\theta,a)$ as a PWLC function of $p_1$. 
}
\end{figure} 

For a high-dimensional state space  $N\geq 3$, the user's entire belief region  $\Delta \mathcal{X}$ is an $N-2$ simplex. For each type $\theta$, we can divide the entire belief region into at most $K$ sub-regions $\mathcal{C}^{\theta}_{a_i}:=\{\mathbf{p}\geq \mathbf{0} | \mathbf{p}' [\mathbf{\hat{v}}_U(\theta, a_i)-\mathbf{\hat{v}}_U(\theta, a_j)] \geq 0, \forall a_j\in \mathcal{A}$. 
Then, $\Delta \mathcal{X}=\cup_{i\in \{\textcolor{black}{DO},1,\cdots,K-1\}} \mathcal{C}_{a_i}^{\theta}$. 
If the posterior belief  falls into the sub-region $\mathcal{C}_{a_i}^{\theta}$,  the user of type $\theta$ takes $a_i$ as his best-response action. 
Take Fig. \ref{fig:binaryPWLC} as an example,  $\mathcal{C}_{a_{DO}}^{\theta}$ is the interval $[t_2^{\theta},1]$ and  $\mathcal{C}_{a_2}^{\theta}$ is the empty set. 
As a direct result of the definition of convexity,  
sets $\mathcal{C}_{a_i}^{\theta}, \forall i\in \{\textcolor{black}{DO},1,\cdots, K-1\}$, are convex and connected. 

We have illustrated the belief region partition under any given type $\theta\in \Theta$. 
Since the user has $M$ possible types, we further divide the belief region into finer sub-regions. 
Let $\mathcal{C}_{\{a^1,\cdots,a^M\}}:=\mathcal{C}_{a^1}^{\theta_1} \cap \cdots \cap \mathcal{C}_{a^M}^{\theta_M}$ be the sub-region of the posterior belief under which the best-response action of the 
user of type $\theta_l, \forall l \in \{1,\cdots,M\}$, is action $a^l\in \mathcal{A}$. 
In particular, define $\mathcal{C}_{i,j}^{l,h}:=\mathcal{C}_{a_i}^{\theta_l} \cap \mathcal{C}_{a_j}^{\theta_h}$ as the belief region where the user takes action $a_i$ when his type is $\theta_l$ and $a_j$ when his type is $\theta_h$ for all $i,j\in \{\textcolor{black}{DO},1,\cdots,K-1\}$ and $l\neq h, \forall l,h\in \{1,\cdots,M\}$. 
\textcolor{black}{Based on the definition, $\mathcal{C}_{i,j}^{l,h}\equiv \mathcal{C}_{j,i}^{h,l}$.} 
Since the intersection of any collection of convex sets is convex, $\mathcal{C}_{\{a^1,\cdots,a^M\}}$ and $\mathcal{C}_{i,j}^{l,h}$ are all convex and connected sets, i.e.,  convex polytopes. 
\textcolor{black}{
We visualize these convex polytopes in Fig. \ref{fig:partision} when there are two types $M=2$, two actions $K=2$, and three states $N=3$. The belief region $\Delta \mathcal{X}$ is an $N-2$ simplex, i.e., an equilateral triangle. 
Under type $\theta_1$, the belief region is divided into $\mathcal{C}_{a_{DO}}^{\theta_1}=\mathcal{C}_{\{a_{DO},a_1\}}\cup \mathcal{C}_{\{a_{DO},a_{DO}\}}$ and $\mathcal{C}_{a_1}^{\theta_1}=\mathcal{C}_{\{a_{1},a_1\}}\cup \mathcal{C}_{\{a_{1},a_{DO}\}}$. 
Under type $\theta_2$, the belief region is divided into $\mathcal{C}_{a_{DO}}^{\theta_2}=\mathcal{C}_{\{a_{DO},a_{DO}\}}\cup \mathcal{C}_{\{a_{1},a_{DO}\}}$ and $\mathcal{C}_{a_1}^{\theta_2}=\mathcal{C}_{\{a_{1},a_1\}}\cup \mathcal{C}_{\{a_{DO},a_{1}\}}$. 
Since there are only two types, we have $\mathcal{C}_{1,DO}^{1,2}=\mathcal{C}_{\{a_{1},a_{DO}\}}$. 
}
\begin{figure}[h]
\centering
\includegraphics[width=.85\columnwidth]{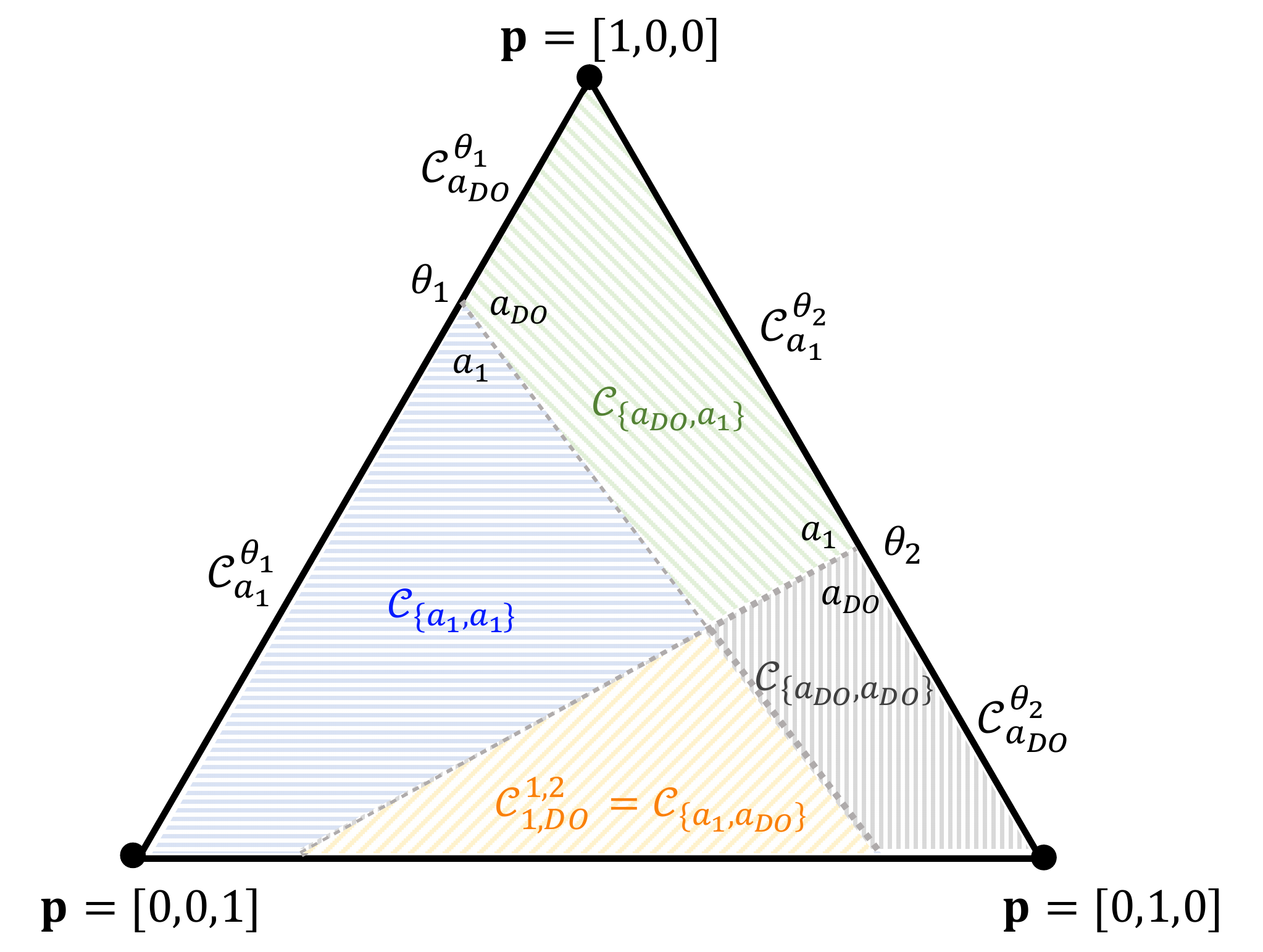}
\caption{ \label{fig:partision}
\textcolor{black}{
Illustration of  $K^M=4$ convex polytopes $\mathcal{C}_{\{a_1,a_1\}}$, $\mathcal{C}_{\{a_{DO},a_1\}}$, $\mathcal{C}_{\{a_{DO},a_{DO}\}}$, and $\mathcal{C}_{\{a_1,a_{DO}\}}$ in blue (horizontal stripes), green (downward diagonal stripes), grey (vertical stripes), and orange (upward diagonal stripes), respectively. 
Each point in the equilateral triangle represents a belief $\mathbf{p}=[p_1,p_2,p_3]\in \Delta\mathcal{X}$. 
}
}
\end{figure}


Among $K^M$ possible sets  $\mathcal{C}_{\{a^1,\cdots,a^M\}}, \forall a^l\in \mathcal{A},l\in \{1,\cdots,M\}$, most of them are empty. 
Take $N=2$ as an example,  $K$ actions can generate at most $K(K-1)/2$ belief thresholds over $p_1\in (0,1)$ for each type as shown in Fig. \ref{fig:binaryPWLC}. Thus, the whole belief region $p_1\in [0,1]$ can be divided into at most $MK(K-1)/2+1$ regions under $M$ types. 
When $N=3$, the belief region is an equilateral triangle  \textcolor{black}{as shown in Fig. \ref{fig:partision}.} 
For each given type, $K$ actions represent $K$ planes. 
Projecting these planes vertically onto the equilateral triangle, we obtain at most $K(K-1)/2$ lines. 
Thus, these lines under $M$ types can divide the equilateral triangle into at most $\frac{MK(K-1)}{2}(\frac{MK(K-1)}{2}+1)/2$ belief regions. 
The results can be extended to $N>3$ as a variant of the hyperplane arrangement problem \cite{orlik2013arrangements}. 
\textcolor{black}{We summarize the above result in Proposition \ref{proposition:DivisionRestricts}; i.e.,} the number of belief region partitions grows in a polynomial rate denoted by $\chi (K,M,N)$ rather than the exponential rate of $K^M$, where $\chi (K,M,N)$ is a polynomial function of $K,M$ for each $N$.  

\begin{proposition} [\textbf{Upper Limit of Enforceable Policies}]  
\label{proposition:DivisionRestricts}
For any credible generator,  
at most $\chi (K,M,N)$ security policies are enforceable. 
\end{proposition}
\begin{remark}
\textcolor{black}{Solely dependent on the user's utility vector $\mathbf{\hat{v}}_U$, the belief partition $\Delta \mathcal{X}=\cup_{a^1\in \mathcal{A},\cdots,a^M\in \mathcal{A}} \mathcal{C}_{\{a^1,\cdots,a^M\}}$ characterizes the user's incentive under different types.}   
If $\mathcal{C}_{\{a^1,\cdots,a^M\}}=\emptyset$, then the security policies that require the user of type $\theta_l$ to take action $a^l$ for any $l\in \{1,\cdots,M\}$ are unenforceable as they violate the user's incentive. 
Proposition \ref{proposition:DivisionRestricts} illustrates that the number of enforceable security policies cannot exceed a threshold determined by $K,M,N$; 
i.e., among all $|\mathcal{S}|=K^M$ potential security policies, the defender can choose at most $\chi (K,M,N)$ ones to be compatible with the user's incentive. 
\end{remark}

\subsubsection{Cyber Attribution and Type Identification} 
The honeypot example motivates us to investigate the condition under which public security policies elicit different actions from different types of users. 
The condition is useful for cyber attribution, i.e., tracing observable actions back to the user's private types. 
Since each security policy uniquely determines a posterior belief for a given generator, we define type identifiability concerning the posterior belief in Definition \ref{def:separable}. 

\begin{definition}[\textbf{Identifiable Types}]
\label{def:separable}
Two different types $l,h \in \{1,\cdots,M\}$ are identifiable under a posterior belief $\mathbf{p}\in \Delta \mathcal{X}$ 
if $ \exists i,j\in \{\textcolor{black}{DO},1,\cdots,K-1\}$ and $i\neq j$ such that $\mathbf{p}\in  \mathcal{C}_{i,j}^{l,h}$.
\end{definition}
The posterior beliefs under which two different types $l,h \in \{1,\cdots,M\}$ are identifiable constitute a belief region that may not be connected. 
\textcolor{black}{This belief region solely depends on the user's utility vector $\mathbf{\hat{v}}_U$ as the finest belief partition $\Delta \mathcal{X}=\cup_{a^1\in \mathcal{A},\cdots,a^M\in \mathcal{A}} \mathcal{C}_{\{a^1,\cdots,a^M\}}$ solely depends on $\mathbf{\hat{v}}_U$.} 
Intuitively, the size of the region is reduced as the utilities of the users of type $\theta_l$ and $\theta_h$ become better aligned. 
Definition \ref{def:perfectalignedUility} defines two extremes of utility alignment. 

\begin{definition}[\textbf{Completely (Mis)aligned Utilities}]
\label{def:perfectalignedUility}
Two different types of users have completely aligned (resp. misaligned) utilities, or equivalently zero (resp. full) utility misalignment, if they are unidentifiable (resp. identifiable) under all posterior belief $\mathbf{p}\in \Delta \mathcal{X}$.   
\end{definition}
If two utilities have the same (resp. opposite) values, then they are completely aligned (resp. misaligned). 
If two types of users' utilities are completely aligned (resp. misaligned), then the security policies that procure them to take different actions (resp. the same action) are not enforceable under any credible generators. 
 Proposition \ref{prop:alignment of users} shows that the results are translation- and scale-invariant. 
\begin{proposition}[\textbf{Alignment under Linear Dependence}]
\label{prop:alignment of users}
Consider linearly dependent utilities of two types $l,h \in \{1,\cdots,M\}$ of users; i.e., there exist a scaling factor $\rho_U^s(\theta_l,\theta_h)\in \mathbb{R}$ and translation factors $\rho_U^t(x, \theta_l,\theta_h) \in \mathbb{R}, \forall x\in \mathcal{X}$, such that  $\hat{v}_U(x,\theta_l,a)=\rho_U^s(\theta_l,\theta_h)\hat{v}_U(x,\theta_h,a)+\rho_U^t(x,\theta_l,\theta_h), \forall x\in \mathcal{X},a\in \mathcal{A}$. 
Two utilities are completely aligned (resp. misaligned) if and only if $\rho_U^s(\theta_l,\theta_h)\geq 0$ (resp. $<0$). 
\end{proposition}

\begin{proof}
For any given $\mathbf{p}\in \Delta \mathcal{X}$ and $\theta_l\in \Theta$, there exists an action $a^*_i\in \mathcal{A}$ such that $\sum_{n=1}^N p_n \allowbreak
[\hat{v}_U(x_n,\theta_l,a^*_i)\allowbreak
-\hat{v}_U(x_n,\theta_l,a_k)]\geq 0, \allowbreak
\forall a_k\in\mathcal{A}$. 
Then, $ \rho_U^s(\theta_l,\theta_h) \sum_{n=1}^N p_n [\hat{v}_U(x_n,\theta_h,a^*_i)-\hat{v}_U(x_n,\theta_h,a_k)]\geq 0, \allowbreak
\forall a_k\in\mathcal{A}$, and the user of type $\theta_h\in \Theta$ at any posterior belief $\mathbf{p}$ has the same best-response action $a^*_i$ if and only if $\rho_U^s(\theta_l,\theta_h)\geq 0$. 
\end{proof}



\subsubsection{Characterization of the Optimal Generator}
\label{sec:convex hull}
Under a zero-information generator $\pi^0\in \Pi$, the user's posterior belief equals the prior belief $\mathbf{p}^0$ and we can rewrite  the user' best-response action $a^*_{\theta}(b_U^{\pi^0})$ in \eqref{eq:optimalAction} as $a_{\theta}^*(\mathbf{p}^0)$. 
Since variables $b_U,c$ are not designable in the benchmark case, we omit them in function $\bar{v}_D$ and rewrite the defender's expected posterior utility as $\bar{v}_D(\pi,\mathbf{p}^0)$. 
Since the users make decisions based on their prior beliefs, we refer to the expected posterior utility $\bar{v}_i$ of player $i\in \{D,U\}$ as his \textit{prior utility} $\tilde{v}_i$ when the generator contains zero information. 
In particular, the defender's prior utility $\tilde{v}_D$ is a function of  the prior belief $\mathbf{p}^0$, i.e., 
\begin{equation*}
    \tilde{v}_D(\mathbf{p}^0):=\bar{v}_D(\pi^0,\mathbf{p}^0)=\mathbb{E}_{x\sim \mathbf{p}^0} \mathbb{E}_{\theta\sim b_D(\cdot|x)} [ \hat{v}_D(x,\theta,a_{\theta}^*(\mathbf{p}^0)].  
\end{equation*}
\textcolor{black}{We obtain the piece-wise linear structure of the defender's prior utility $\tilde{v}_D$ in Proposition \ref{proposition:senderPWL}. The solid lines in Fig. \ref{fig:utilityhacking} illustrate $\tilde{v}_D$.} 

\begin{proposition}
\label{proposition:senderPWL}
The defender's prior utility $\tilde{v}_D$ is a (possibly discontinuous) piece-wise linear function of the common prior belief vector $\mathbf{p}^0\in \Delta\mathcal{X}$ with at most $\chi (K,M,N)$ pieces. 
\end{proposition}
\begin{proof}
\textcolor{black}{The piece-wise linear structure follows from the fact that $\tilde{v}_D$ is linear with respect to $\mathbf{p}^0$ inside each convex polytope $ \mathcal{C}_{\{a^1,\cdots,a^M\}}, \forall a^l\in \mathcal{A}, l\in \{1,\cdots,M\}$. As a result of Proposition \ref{proposition:DivisionRestricts}, the upper bound of the number of different convex polytopes is $\chi (K,M,N)$. 
Since the polytopes are determined based on the user's prior utility rather than the defender's, $\tilde{v}_D$ is possibly discontinuous at the boundaries of these polytopes.} 
\end{proof}

\begin{figure}[h]
\centering
\includegraphics[width=1\columnwidth]{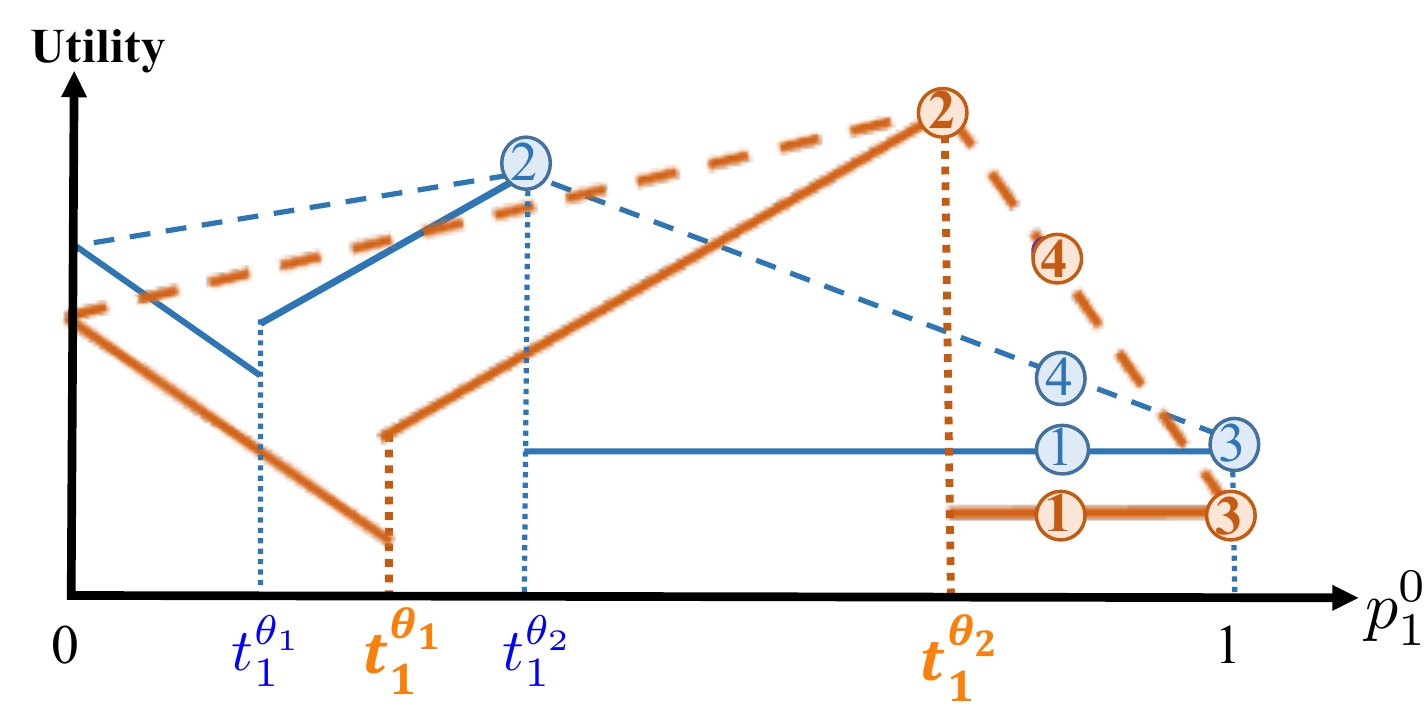}
\caption{ 
The defender's expected posterior utility versus prior belief $p_1^0$ 
with and without the modulator in orange and blue, respectively.  
\textcolor{black}{We denote orange lines and notations in bold.}
The solid lines indicate that the defender's prior utility $\tilde{v}_D$ is discontinuous and  piece-wise linear  under three belief regions, i.e.,  $[0,t_1^{\theta_1}],[t_1^{\theta_1},t_1^{\theta_2}]$, and $[t_1^{\theta_2},1]$.  
The dashed lines represent the defender's optimal posterior utility $V_D$. 
}
\label{fig:utilityhacking}
\end{figure}

The defender's expected posterior utility $\bar{v}_D$ is a function of $\pi\in \Pi$ and $\mathbf{p}^0\in \Delta\mathcal{X}$. 
Thus, the defender's optimal posterior utility $V_D(\mathbf{p}^0):=\sup_{\pi\in\Pi} \bar{v}_D(\pi,\mathbf{p}^0)$ is a function of $\mathbf{p}^0\in \Delta\mathcal{X}$. 
Based on Theorem \ref{thm:feasibleBounded}, there exists an optimal generator $\pi^*\in\Pi$ that achieves the optimal posterior utility, i.e., $V_D(\mathbf{p}^0)= \bar{v}_D(\pi^*,\mathbf{p}^0)=r$. 
Denote the convex hull of function $\tilde{v}_D$ as $co(\tilde{v}_D)$. 
Then, we can use the concavification technique introduced in \cite{aumann1995repeated, kamenica2011bayesian} to show that the defender's optimal posterior utility $V_D(\mathbf{p}^0)$ is the concave closure of her prior utility $\tilde{v}_D(\mathbf{p}^0)$ over the entire belief region $\mathbf{p}^0\in \Delta \mathcal{X}$, i.e., $V_D(\mathbf{p}^0)=\sup \{z\in \mathbb{R} | (\mathbf{p}^0,z)\in co(\tilde{v}_D)\}$. 



We visualize the concavification process under the binary state space $N=2$ in Fig. \ref{fig:utilityhacking}.  
 Suppose that there are two types of users and each type $\theta\in \{\theta_1,\theta_2\}$ has a single belief threshold denoted by $t_1^{\theta}$ \textcolor{black}{where $0<t_1^{\theta_1}<t_1^{\theta_2}<1$.}   
Consider a common prior belief $p^0_1\in [t_1^{\theta_2},1]$ denoted by node $1$'s abscissa. 
Then, the defender's prior utility $\tilde{v}_D(p_1^0)$ is denoted by node $1$'s ordinate. 
The defender can improve the utility from node $1$'s ordinate to at most node $4$'s ordinate by adopting the optimal generator $\pi^*\in \Pi$ as follows. 
Generator $\pi^*$ generates two signals $s_{2}\in \mathcal{S}$ and $s_{3}\in \mathcal{S}$ with proper probabilities under different states so that the user's posterior belief is node $2$'s abscissa when observing policy $s_{2}$ and node $3$'s abscissa when observing $s_{3}$. Based on the Bayesian plausibility condition \textcolor{black}{in Section \ref{sec:relation2BP},}  
the defender's optimal posterior utility $V_D(p_1^0)$ can be represented as the linear interpolation of the ordinates of nodes $2$ and $3$, i.e., node $4$'s ordinate. 
The same reasoning applies to all feasible common prior beliefs $p^0_1\in [0,1]$. 
Therefore, for all $[p_1^0,1-p_1^0] \in \Delta \mathcal{X}$, the defender's optimal posterior utility $V_D(\mathbf{p}^0)$ is the concave closure of her prior utility $\tilde{v}_D(\mathbf{p}^0)$ and $V_D(\mathbf{p}^0) \geq \tilde{v}_D(\mathbf{p}^0)$. 

Although we need at least $|\mathcal{S}|=K^M$ security policies to represent all the permutations of actions under different types, Fig. \ref{fig:utilityhacking} shows that the defender can achieve her optimal posterior utility by generating two different security policies with proper probabilities when $N=2$. 
Proposition \ref{proposition: OnlyTwoSignalSent} generalizes the result to $N>2$ and shows that the generator only needs to generate a small number of security policies to achieve her optimal posterior utility. 
If $\tilde{v}_D(\mathbf{p}^0)=V_D(\mathbf{p}^0)$ and $\mathbf{p}^0$ is further an interior point of any
convex polytope  $\mathcal{C}_{\{a^1,\cdots,a^M\}}, \forall a^l\in\mathcal{A},l\in\{1,\cdots,M\}$, then there exist infinitely many credible generators that achieve $V_D(\mathbf{p}^0)$. 

\begin{proposition}[\bf{Efficiency of the Optimal Generator}]
\label{proposition: OnlyTwoSignalSent}
For any DG with common prior belief $\mathbf{p}^0\in \Delta \mathcal{X}$, there exist either one or infinitely many optimal generators to achieve the optimal posterior utility  $V_D(\mathbf{p}^0)$. 
For each state $x\in \mathcal{X}$, there exists one optimal generator $\pi^*(\cdot|x)\in \Delta \mathcal{S}$ that generates at least $K^M-N$ security policies with zero probability. 
\end{proposition}

\begin{proof}
Since COP under the benchmark case is a linear program, the optimal solution is either unique or innumerable. 
If $N=2$, the convex hull consists of pieces of line segments where each line segment can be determined uniquely by its two endpoints. 
If $N=3$, the convex hull as a polygon consists of finite pieces of triangles  where each triangle can be determined uniquely by its three endpoints. 
We can extend to any finite $N$ where the convex hull consists of pieces of $(N-1)$-simplex where each piece can be determined uniquely by $N$ endpoints. 
Thus, for any $\mathbf{p}^0\in \Delta\mathcal{X}$, it requires at most $N$ points to achieve $V_D(\mathbf{p}^0)$, which corresponds to $N$ distinct security policies. 
\end{proof}

\begin{remark} 
Proposition \ref{proposition: OnlyTwoSignalSent}  shows that the defender does not need to apply all enforceable security policies to achieve the optimal posterior utility; i.e., the optimal generator is efficient and generates at most $N$ security policies for each state $x\in \mathcal{X}$.
\label{remark:onlyTwoSignal}
\end{remark}


We define the trust margin under a credible generator $\pi\in \Pi$ in Definition \ref{def:trust margin}. 
 The maximum trust margin is achieved when the optimal generator $\pi^*\in \Pi$ is applied. 
 The trust margin can be negative if generator $\pi$ is not well designed. 
 However, the maximum trust margin is non-negative as it is the difference between the defender's optimal posterior utility and prior utilities, i.e.,  $V_D(\mathbf{p}^0)-\tilde{v}_D(\mathbf{p}^0)$. 
 Based on whether the maximum trust margin is zero or positive, Definition \ref{def:manageable user} defines the user to be unmanageable or manageable. 

\begin{definition}[\textbf{Trust Margin}]
\label{def:trust margin}
We define $\bar{v}_D(\pi,\mathbf{p}^0)-\tilde{v}_D(\mathbf{p}^0)$ as the  trust margin under the common prior belief $\mathbf{p}^0\in \Delta \mathcal{X}$ and a credible generator $\pi\in \Pi$. 
\end{definition}
\begin{definition}[\textbf{Manageability}]
\label{def:manageable user}
The user is manageable (resp. unmanageable) under prior belief $\mathbf{p}^0$ if the maximum trust margin is greater than (resp. equals) zero. 
\end{definition}
Intuitively, a user is manageable if he shares the same utility with the defender but unmanageable if he has an opposite utility. 
We introduce $\rho^s_D\in \mathbb{R}$ to represent the user's level of maliciousness. 
Theorem \ref{thm:Alignment} investigates how the user's level of maliciousness affects his manageability. 


\begin{theorem}[\textbf{Manageability and Level of maliciousness}]
\label{thm:Alignment}
Let the common prior belief be state-independent, i.e., $b_D(\theta|x) \allowbreak
=\hat{b}_D(\theta), \forall \theta\in \Theta, \forall x\in \mathcal{X}$, and two players' utilities be linearly dependent, i.e., there exist a scaling factor $\rho^s_D \in \mathbb{R}$ and translation factors $\rho^t_D(x, \theta) \in \mathbb{R}$, such that 
$
\hat{v}_D(x,\theta,a)=\rho^s_D \hat{v}_U(x,\theta,a)+\rho^t_D(x, \theta), \forall  x\in \mathcal{X}, \theta\in \Theta, a \in \mathcal{A}$. 
Then, the following two statements hold. 

\begin{itemize}
\item[(a)]  
The defender's trust margin is zero for all $\mathbf{p}^0\in \Delta \mathcal{X}$ and credible generators if and only if $\rho^s_D\leq 0$. The optimal generator contains zero information. 
\item[(b)] 
The defender's trust margin is non-negative for all $\mathbf{p}^0\in \Delta \mathcal{X}$ and credible generators if and only if $\rho^s_D>0$. Moreover, the optimal generator contains full information. 
If  $\mathbf{p}^0$ is an interior point of the $(N-1)$-simplex and there exists at least one $\theta\in \Theta$ under which  no actions dominate, then the defender's trust margin is positive. 
\end{itemize}
\end{theorem} 



\begin{proof}
Under the given conditions, 
$
 \tilde{v}_D(\mathbf{p}^0)   
= 
\mathbb{E}_{\theta\sim \hat{b}_D} \allowbreak
\mathbb{E}_{x\sim \mathbf{p}^0} \allowbreak
[ \rho^s_D \hat{v}_U(x,\theta,a_{\theta}^{*}(\mathbf{p}^0))
+\rho^t_D (x, \theta)]
=\rho^s_D \mathbb{E}_{\theta\sim \hat{b}_D} \allowbreak
\mathbb{E}_{x\sim \mathbf{p}^0} \allowbreak
[  \hat{v}_U(x,\theta,a_{\theta}^{*}(\mathbf{p}^0))]
+\mathbb{E}_{\theta\sim \hat{b}_D} \allowbreak
\mathbb{E}_{x\sim \mathbf{p}^0}\allowbreak
[\rho^t_D(x, \theta)]. 
$
Proposition \ref{proposition:PWLC} has shown that $\mathbb{E}_{x\sim \mathbf{p}^0} [  \hat{v}_U(x,\theta,a_{\theta}^*(\mathbf{p}^0))]$ is a PWLC function of $ \mathbf{p}^0$ for each $\theta\in \Theta$.  Since $\hat{b}_D(\theta)\geq 0, \forall \theta\in \Theta$, the linear combination $\mathbb{E}_{\theta\sim \hat{b}_D}   \mathbb{E}_{x\sim \mathbf{p}^0} [  \hat{v}_U(x,\theta,a_{\theta}^*(\mathbf{p}^0))]$ is also PWLC. 
The term $\mathbb{E}_{\theta\sim \hat{b}_D}  \mathbb{E}_{x\sim \mathbf{p}^0} [\rho^t_D (x, \theta)]$ is a linear function of $ \mathbf{p}^0$. 
Thus, $\tilde{v}_D$ is a piece-wise linear and concave (resp. linear) function of $ \mathbf{p}^0$ if and only if $\rho^s_D<0$ (resp. $\rho^s_D=0$). 
If  $\tilde{v}_D$ is concave or linear over the entire belief region $ \Delta \mathcal{X}$, its convex hull is itself. 
Thus, $V_D(\mathbf{p}^0)=\tilde{v}_D(\mathbf{p}^0)$ for all $ \mathbf{p}^0\in \Delta \mathcal{X}$ and any zero-information generator is optimal. 
Similarly, $\tilde{v}_D$ is PWLC  if and only if $\rho^s_D>0$, and any full-information generator is optimal. 
If there exists at least one $\theta\in \Theta$ under which no actions dominate, then $\tilde{v}_D$ is strictly convex over  the entire belief region. 
Thus, we have $V_D(\mathbf{p}^0)<\tilde{v}_D(\mathbf{p}^0)$ when $\mathbf{p}^0$ is an interior point of the $(N-1)$-simplex. 
\end{proof}

Theorem \ref{thm:Alignment} shows that when two players' utilities are linearly dependent, the user's manageability depends on the sign of the scaling factor $\rho^s_D$ rather than its value. 
Thus, the user's level of maliciousness has a threshold impact on the manageability and the threshold is $0$. 

\subsection{Incentive Modulator and Trust Manipulator}
\label{subsec:Modulator}
We illustrate the modulator design and the manipulator design
in Section \ref{subsec:ModulatorandGenerator} and \ref{subsec:Manipulatorandgenerator}, respectively. 
The GMM mechanism design is presented in Section \ref{sec:all together}. 
\subsubsection{Joint Design of Generator and Modulator}
\label{subsec:ModulatorandGenerator}
The modulator incentivizes unmanageable users and increases the security and efficiency of the networks. 
Under the binary state $N=2$, Fig. \ref{fig:utilityhacking} illustrates the defender's prior utility with the modulator in orange solid lines.  
The orange solid lines are different from the blue ones in two folds.  
From the user's perspective, the modulator changes the user's expected utility under different actions and thus  results in translations of the dashed lines in Fig. \ref{fig:binaryPWLC}. Those translations change the belief region partition, e.g., the right shifts of $t_1^{\theta_1}$ and $t_1^{\theta_2}$ in Fig. \ref{fig:utilityhacking}. 
From the defender's perspective, the modulator modifies her utility in each new belief regions, and the value of the modification is $\mathbb{E}_{x\sim \mathbf{p}^0}\mathbb{E}_{\theta\sim b_D(\cdot|x)}[\gamma c(a_{\theta}^*(\mathbf{p}^0))]$. 
If the defender's belief is independent of state, i.e., $b_D(\theta|x)=\hat{b}(\theta), \forall \theta\in \Theta, \forall x\in \mathcal{X}$, then the defender's utility change $\mathbb{E}_{x\sim \mathbf{p}^0}\mathbb{E}_{\theta\sim b_D(\cdot|x)}[\gamma c(a_{\theta}^*(\mathbf{p}^0))]=\gamma \mathbb{E}_{\theta\sim \hat{b}_D(\cdot)}[ c(a_{\theta}^*(\mathbf{p}^0))]$ is a  constant with respect to $\mathbf{p}^0$ in each new belief region. 
When the state space is binary as shown in Fig. \ref{fig:utilityhacking}, it means that designing $c$  introduces translations but not rotations to each segment of the function $\tilde{v}_D$. 

The joint design of the modulator and the generator results in the new convex hull denoted by the dashed blue lines in Fig. \ref{fig:utilityhacking}.  
Based on both players' perspectives, the optimal design needs to strike a balance between incentivizing users to change their belief region partitions and the costs to provide the incentives.  
Take Fig. \ref{fig:utilityhacking} as an example, we observe that the modulator incurs costs to the defender for all actions, i.e.,  $c(a)\leq 0, \forall a\in \mathcal{A}$. 
Thus, in all three belief regions, the defender's prior utilities with the modulator, represented by the solid orange lines, are lower than the ones without the modulator, represented by the solid blue lines. 
However, the benefit of the user's incentive change outweighs the costs; i.e., the defender's optimal posterior utility $V_D(p_1^0)$ increases from node  $4$ in blue to node $4$ in orange. 
\subsubsection{Joint Design of Generator and Manipulator}
\label{subsec:Manipulatorandgenerator}
The manipulator directly distorts the user’s prior belief to elicit desirable behaviors.
When the generator cannot be designed, the manipulator design is equivalent to the process of finding the initial belief $\mathbf{p}^0_g:=arg\max_{\mathbf{p}^0\in  \Delta \mathcal{X}}\tilde{v}_D(\mathbf{p}^0)$ that achieves the global maximum of the prior utility $\tilde{v}_D$. 
Proposition \ref{proposition:optimal distorted belief} proves the existence of the optimal distorted belief $\mathbf{p}^0_g$. 
\begin{proposition}
\label{proposition:optimal distorted belief}
For any given $\hat{v}_D,\hat{v}_U$ of two players, there exists an initial belief $\mathbf{p}^0_g\in \Delta \mathcal{X}$ at the boundary of the convex polytopes $ \mathcal{C}_{\{a^1,\cdots,a^M\}}, \forall a^l\in \mathcal{A}, l\in \{1,\cdots,M\}$, 
such that $\mathbf{p}^0_g=arg\max_{\mathbf{p}^0\in  \Delta \mathcal{X}}\tilde{v}_D(\mathbf{p}^0)$. 
\end{proposition}
\begin{proof}
For each $\hat{v}_D,\hat{v}_U$, the global maximum  $\tilde{v}_D(\mathbf{p}^0_g)=\max_{\mathbf{p}^0\in  \Delta \mathcal{X}}\tilde{v}_D(\mathbf{p}^0)$ exists and has a finite value due to Theorem \ref{thm:feasibleBounded}. 
Proposition \ref{proposition: OnlyTwoSignalSent} shows that the global maximum is either unique or infinite. In either case, at least one global maximum is at the boundary of the convex polytopes due to the piece-wise linear property stated in Proposition \ref{proposition:senderPWL}.  
\end{proof}

When the optimal generator is applied, the joint design of the manipulator and the generator is equivalent to the process of finding the initial belief $\mathbf{\bar{p}}^0_g:=arg\max_{\mathbf{p}^0\in  \Delta \mathcal{X}}V_D(\mathbf{p}^0)$ that achieves the global maximum of $V_D$. 
Based on the  piece-wise linear property of $\tilde{v}_D$ in Proposition \ref{proposition:senderPWL}, 
the prior utility $\tilde{v}_D$ and its concave closure $V_D$ share the same global maximum. 
Thus, $\mathbf{p}^0_g=\mathbf{\bar{p}}^0_g$ and the optimal generator contains zero information. 
Take Fig. \ref{fig:utilityhacking} as an example, $\mathbf{p}^0_g=[t_1^{\theta_2},1-t_1^{\theta_2}]$ achieves the global maximum denoted by node $2$'s ordinate, and node $2$ is on both the solid and the dashed lines. 
These results are summarized in Theorem \ref{thm:InitialbeliefManipulation}. 
\begin{theorem}
\label{thm:InitialbeliefManipulation}
The design of optimal overt manipulator changes the common initial belief $\mathbf{p}^0$ into $\mathbf{p}^0_g=\mathbf{\bar{p}}^0_g$. 
The defender's  optimal posterior utility has the value of $\tilde{v}_D(\mathbf{p}^0_g)=V_D(\mathbf{p}^0_g)$ and is independent of the initial belief $\mathbf{p}^0\in \Delta \mathcal{X}$.  
In the joint design of the overt manipulator and the generator, the optimal generator contains zero information. 
\end{theorem}

\subsubsection{Design of the GMM Mechanism}
\label{sec:all together}
We incorporate the modulator design into the joint design of the generator and the manipulator to complete the GMM mechanism design. 
Based on the analysis in Section \ref{subsec:Manipulatorandgenerator}, 
the first step of the GMM design is to determine the optimal modulator $c^*\in \mathcal{C}$ that results in the prior utility function with the largest value of the global maximum, i.e., $c^*=arg\max_{c}[\max_{\mathbf{p}^0\in  \Delta \mathcal{X}}\tilde{v}_D(\mathbf{p}^0)]$.  
With the given modulator $c^*$, the second step of the design is to reduce the problem to the joint design of modulator and manipulator presented in Section \ref{subsec:Manipulatorandgenerator}. 
\begin{remark}[\bf{Separation Principle}] 
\label{remark: Separable Design}
The two-step design of the GMM mechanism shows that the defender can 
design the optimal modulator $c^*\in \mathcal{C}$ independently. 
\end{remark}
\textcolor{black}{
We identify the \textit{equivalence principle} in Remark \ref{remark: Local and Global } based on the results in Theorem \ref{thm:InitialbeliefManipulation}.}  
If the overt manipulator allows the defender to manipulate the initial belief arbitrarily, then the optimal generator contains zero information; i.e., the defender no longer needs the optimal generator to achieve her optimal posterior utility. 
\textcolor{black}{
Note that the equivalence principle does not mean that the generator is redundant. 
When the belief manipulation is not arbitrary and under practical constraints (e.g., the belief changes within a limited range), the joint design of the two components can yield better performance than the single design of the manipulator.} 

\begin{remark}[\bf{Equivalence Principle}] 
\label{remark: Local and Global }
For any given modulator $c\in \mathcal{C}$, the joint design of the generator and the overt manipulator results in the same outcomes as the single design of the overt manipulator does.  
\end{remark}


\section{Case Study}
\label{sec:case study}
In Section \ref{sec:case study}, we illustrate how the defender can use the DG to mitigate insider threats where honeypots are configured adaptively to detect and deter misbehavior. 

\subsection{Model Description}
We have $\Theta=\{\theta^b,\theta^g\}$, $\mathcal{X}=\{x^H, x^N\}$, and $\mathcal{A}=\{a_{DO},a_{AC}\}$ 
\textcolor{black}{
based on the running example introduced in Section \ref{sec:motivatingexample}, Example \ref{example:featurepattern}, and Example \ref{exmple:security policies}.} 
The true percentage of honeypots $p_D^{0,H}:=b(x^H)\in [0,1]$, is only known to the SOC. 
Thus, the insiders' perceived honeypot percentage $p_U^{0,H}:=b_U(x^H|\theta)\in [0,1] , \forall \theta\in \Theta$, can be different from the true percentage. 

Table \ref{table:defender} lists the utilities of the SOC and the insiders. 
The column represents the binary state of a node, and the row represents the insiders' actions. 
In each matrix entry, we list the payoffs resulting from the selfish (resp. adversarial) insiders on the left (resp. right) of the semicolon.
When the insider chooses not to access a node, we calibrate the payoffs to be $0$ for both the SOC and the insiders. 
The other four possible scenarios are listed as follows. 
First, a selfish insider's access to a normal server maintains the organization's normal operation and results in a positive reward $r_D>0$ (resp. $r_U>0$) on average to the organization (resp. the selfish insider). 
Second, when an adversarial insider accesses a normal server, he disrupts the normal operation and compromises confidential data, which brings him a reward of $\phi^{N}_U r_U>0$ and incurs a security loss of $\phi^{N}_D r_D<0$ to the organization. 
Third, if an adversarial insider accesses a honeypot, he is detected and prohibited from data theft. Meanwhile, the SOC obtains valuable threat intelligence. 
We use $\phi^{H}_D>0$ and $\phi^{H}_U<0$ to represent the degrees of the SOC's gain and the adversarial insider's loss, respectively. 
Finally,  once a selfish insider accesses the honeypot, the SOC has to quarantine the insider and investigate the incident, which incurs a suspension of normal services as well as an investigation cost. 
Meanwhile, the selfish insider also receives penalties and additional security training sessions. 
We use $\phi^{g}_D r_D<0$ and $\phi^{g}_U  r_U<0$ to represent the cost for the SOC and the selfish insider, respectively. 
\begin{table}[h]
\centering
\begin{tabular}{|c|c|c|c|}
\hline
Selfish $\theta^g$; Adversarial $\theta^b$ & Honeypot $x^H$      & Normal Server $x^N$      \\ \hline
No Access $a_{DO}$                   & $0$ ; $0$        & $0$ ; $0$        \\ \hline
Access $a_{AC}$                  & $r_i\phi^{g}_i$ ; $r_i\phi^{H}_i$ & $r_i$ ; $r_i\phi^{N}_i$ \\ \hline
\end{tabular}
\caption{Two players' utilities $v_i(x,\theta,a), i\in \{D,U\}$. }
\label{table:defender}
\end{table}

Compared to a computing system that precisely follows its instructions, human insiders alter their behaviors in response to (dis)incentives. 
In this case study, the (dis)incentives refer to the insider's authentication cost $c(a_{AC}):=r_U\phi^0$ to access a node, where the ratio $\phi^0\in \mathbb{R}$ takes the value of $0$ in the default setting. 
We assume that the SOC can increase (i.e., $\phi^0<0$) or decrease (i.e., $\phi^0>0$) an insider's authentication cost at no additional cost, i.e., $\gamma=0$. 
The revenues, losses, and costs can be quantified in dollars and their values vary for different security scenarios. 

\subsubsection{Threshold Policy Analysis}
\label{sec:Threshold Policy Analysis}
In this case study, both selfish and adversarial insiders share the same prior belief $p_U^{0,H}\in [0,1]$. 
Hence they share the same posterior belief denoted by $p_U^H\in [0,1]$ and adopt the following threshold policies. 
Define the decision thresholds of the selfish and the adversarial insiders as
$t^{g}(\phi^0):=\max\{\min\{{(1-\phi^0)}/{(1-\phi_U^g)},1\},0\}$
and $t^{b}(\phi^0):=\max\{\min\{{(\phi_U^N-\phi^0)}/{(\phi_U^N-\phi_U^H)},1\},0\}$, respectively. 
Since both denominators are positive, i.e., $1-\phi_U^g>1$ and $\phi_U^N-\phi_U^H>0$, the selfish insider (resp. the adversarial insider) chooses to access a node if and only if the node is unlikely to be a honeypot, i.e., $p_U^H<t^g(\phi^0)$ (resp. $p_U^H<t^b(\phi^0)$). 
If a selfish (resp. adversarial) insider accesses a node, his expected utility $r_U(1-\phi^0+p_D^{0,H}(\phi_U^g-1))$ (resp. $r_U(\phi_U^N-\phi^0+p_D^{0,H}(\phi_U^H-\phi_U^N))$) decreases linearly in $p_D^{0,H}$, i.e., the true percentage of honeypots.  

Since the selfish and adversarial insiders share the same insider information, the difference in their decision thresholds results purely from their incentive misalignment. 
Given the insiders' utility matrices, the SOC can change their incentives 
and elicit desirable behaviors by a proper design of the authentication cost determined by the ratio $\phi^0$. 
If $\phi^0\leq \phi_U^g<0$ (resp. $\phi^0\leq \phi_U^H<0$), then the selfish (resp. adversarial) insider chooses $a_{AC}$ for all security scenarios. 
If $\phi^0\geq 1$ (resp. $\phi^0\geq \phi_U^N>0$), then the selfish (resp. adversarial) insider  chooses $a_{DO}$ for all security scenarios. 
Since the deceptive honeypot configuration can possibly change insiders' behaviors only if $\phi^0$ is in the region  $[\min(\phi_U^g,\phi_U^H),  \max(1,\phi_U^N) ]$, we refer to the region as the \textit{incentivized region} of $\phi^0$. 
As a special case of Proposition \ref{proposition:DivisionRestricts}, Corollary \ref{corollary:impossible} shows that security policies $s_{\{a_{DO},a_{AC}\}}$ and $s_{\{a_{AC},a_{DO}\}}$ cannot be both enforceable for any node in the corporate network. \begin{corollary}
\label{corollary:impossible}
If $\phi_U^g<0$,$\phi_U^N>0,\phi_U^H<0$, 
then for all $\phi^0\in \mathbb{R}$ and credible configuration $\pi\in \Pi$, either $\pi(s_{\{a_{DO},a_{AC}\}}|x)=0, \forall x\in \{x^H,x^N\}$, or $\pi(s_{\{a_{AC},a_{DO}\}}|x)=0, \forall x\in \{x^H,x^N\}$. 
\end{corollary} 


\subsection{Numerical Results}
\textcolor{black}{Following the insider categorization in Section \ref{sec:Categorization}, we re-weight the percentage from the VCDB and adopt}
$q^g:=b_D(\theta^g|x)=0.32$ and $q^b:=b_D(\theta^b|x)=0.68$ for all $x\in \{x^N,x^H\}$ as the benchmark value of the insiders' type statistics. 
Based on the analysis in Section \ref{sec:Threshold Policy Analysis}, the values of $r_U$ do not affect the insiders' actions, and the value of $r_D$ only scales the SOC's utility by a constant. Thus, we normalize $r_U=r_D=1$. 
We consider $\phi_U^g=\phi_D^g=-0.3$, $\phi_U^H=-\phi_D^H=-1$, and $\phi_U^N=-\phi_D^N=0.9$ as the benchmark values. 
Then, the selfish insider has the same utility as the SOC, i.e., $v_D(x,\theta^g,a)=v_U(x,\theta^g,a),\forall x\in \{x^H,x^N\}, \forall a\in \{a_{AC},a_{DO}\}$, 
while the adversarial insider has an exactly opposite utility to the one of the SOC, i.e., $v_D(x,\theta^b,a)=-v_U(x,\theta^b,a),\forall x\in \{x^H,x^N\}, \forall a\in \{a_{AC},a_{DO}\}$. 
In Section \ref{sec:without configure}, the SOC cannot change the authentication cost, i.e., $c(a_{AC})=0$. 
In Sections \ref{sec:without configure} and \ref{sec:investigateC}, the insider has the correct prior belief of the honeypot percentage, i.e., $p_U^{0,H}=p_D^{0,H}$. 

\subsubsection{Security Posture under the Optimal Generator} 
\label{sec:without configure}

Fig. \ref{fig:Exante} shows how the SOC's normalized revenue $\tilde{v}_D$ without the optimal generator is affected by the percentages of honeypots and the selfish insiders, respectively. 
The maximum (resp. minimum) value of $\tilde{v}_D$ is achieved when insiders are all selfish (resp. adversarial) and no honeypots are applied. 
The two decision thresholds $t^b(0)$ and $t^g(0)$ divide the percentage of honeypots into three regions, i.e., high, medium, and low, in which the insiders' behaviors and the SOC's normalized revenue $\tilde{v}_D$ have different characteristics. 

If the intended security outcomes are not achieved due to the insiders' misbehavior, the SOC can apply the optimal generator to elicit desirable behaviors and reduce the cyber risks of the organization. 
To illustrate the effectiveness of the optimal generator, we plot the maximum trust margin in Fig. \ref{fig:diff}. 
\begin{figure}[h]
    \centering 
    \begin{subfigure}{0.35\textwidth}
\includegraphics[width=1\columnwidth,height=.8\columnwidth]{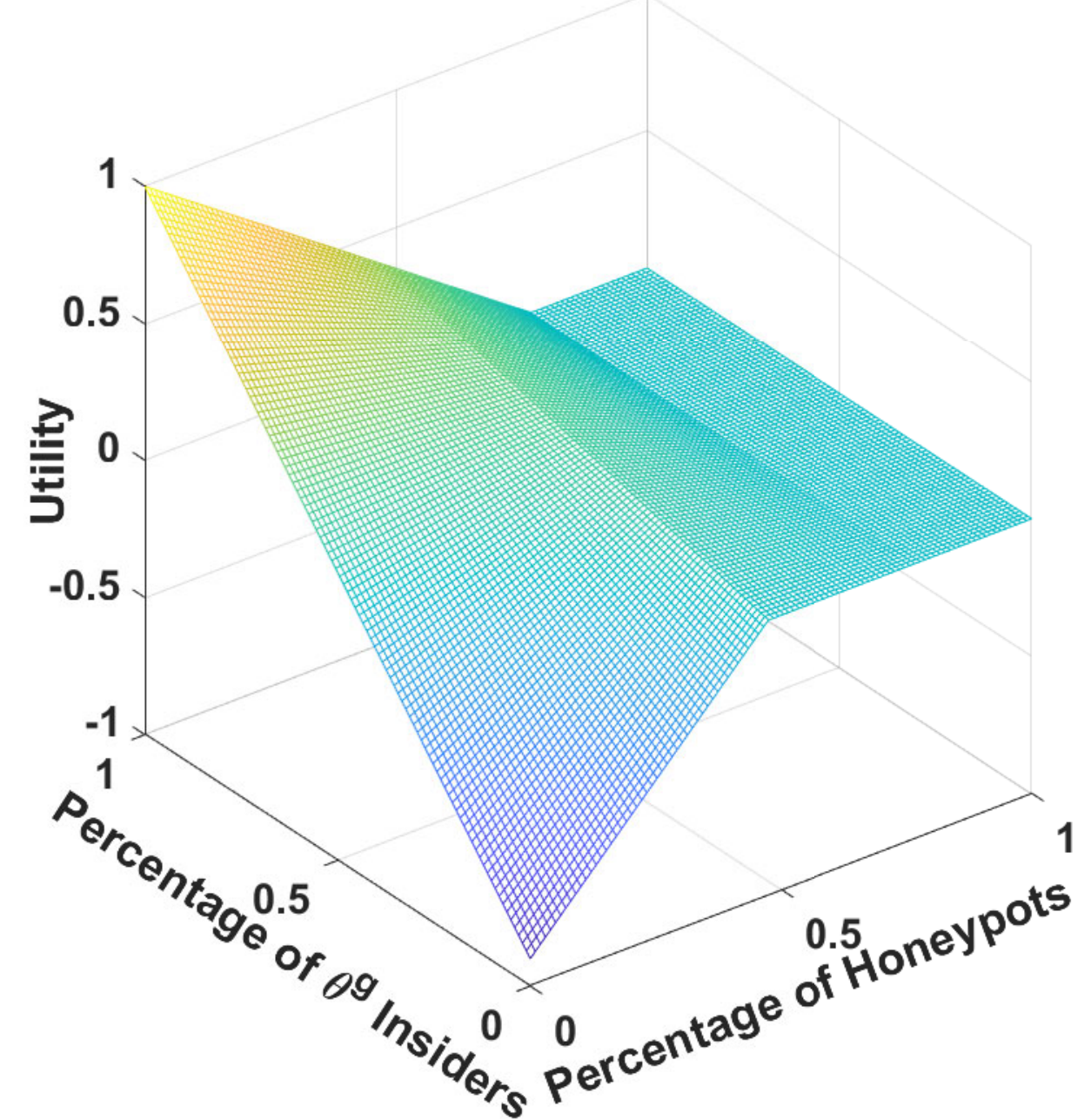}
\caption{ 
Prior utility $\tilde{v}_D$. 
}
\label{fig:Exante}
\end{subfigure}\hfil 
    \begin{subfigure}{0.35\textwidth}
\includegraphics[width=1\columnwidth,height=.8\columnwidth]{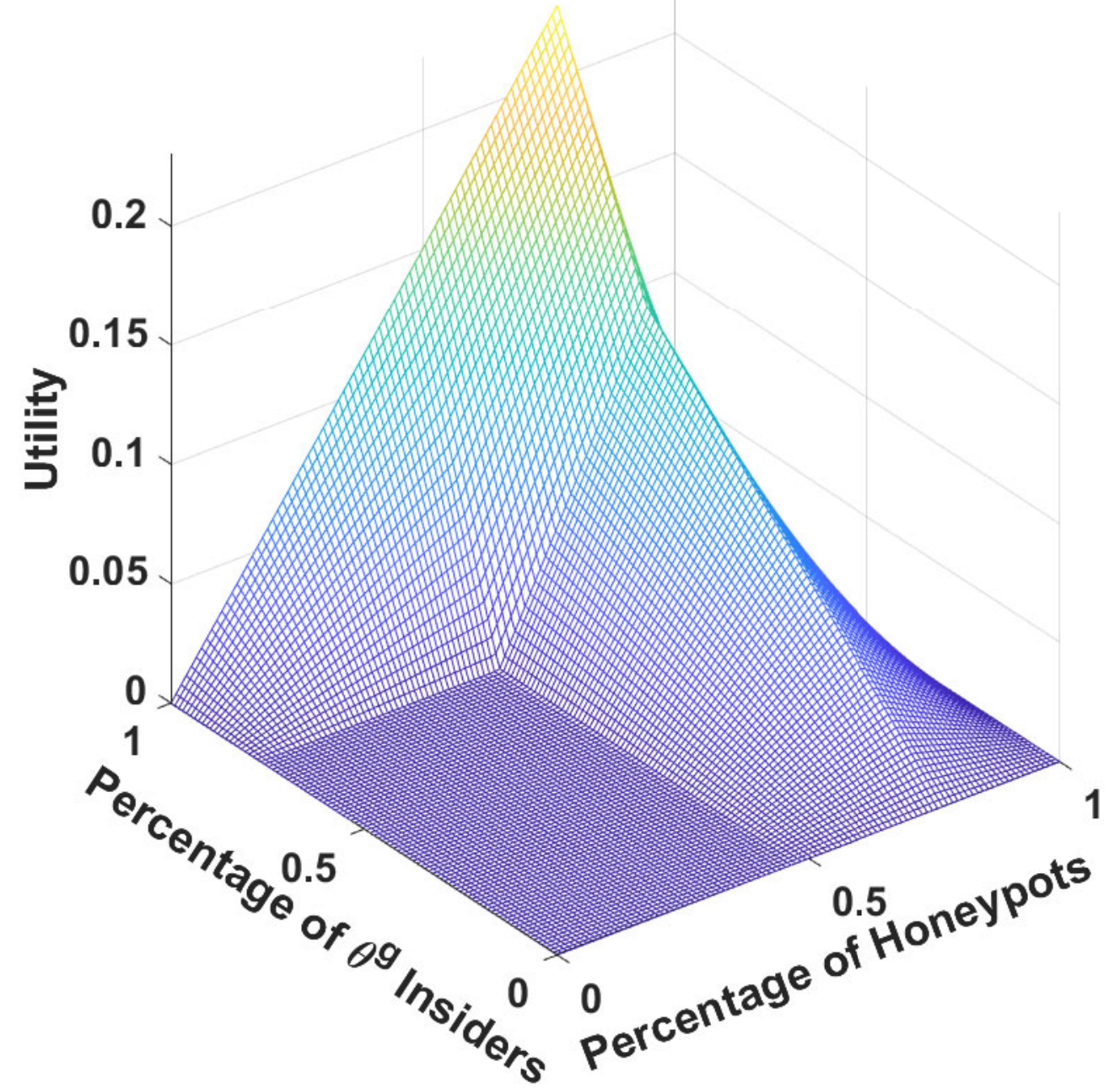} 
\caption{ 
Maximum trust margin. 
}
\label{fig:diff}
\end{subfigure}\hfil 
\caption{
SOC's utilities vs. $p_D^{0,H}\in [0,1]$ and $q^g\in [0,1]$. 
\label{fig:combine}
}
\end{figure}
Fig. \ref{fig:diff} corroborates Theorem \ref{thm:Alignment}; i.e., when all insiders are adversarial (resp. selfish), no (resp. all) credible generators, including the optimal one, can improve the SOC's normalized revenue for any percentage of honeypots $p_D^{0,H}\in [0,1]$. 
The flat region represented by $q^{g}\in [0,{(\phi_D^N-\phi_D^H)}/{(\phi_D^g-1+\phi_D^N-\phi_D^H)}]$ and $p_D^{0,H}\in [0,\min(t^b(0),t^g(0))]$ identifies two critical thresholds. 
On the one hand, we refer to ${(\phi_D^N-\phi_D^H)}/{(\phi_D^g-1+\phi_D^N-\phi_D^H)}$ as the insider's \textit{motive threshold} that is used to quantify the average motive of the entire insider population. 
If the percentage of adversarial insiders exceeds the \textit{motive threshold},  then insiders' behaviors are on average destructive to the organization. 
On the other hand, we refer to $\min(t^b(0),t^g(0))$ as the \textit{deterrence threshold} that measures the adequacy of the honeypots. If the percentage of honeypots is below the \textit{deterrence threshold}, then the SOC does not have a sufficient number of honeypots to create a credible threat for the insiders not to access nodes in the corporate network.  
Based on Definition \ref{def:manageable user}, the insiders are unmanageable in the flat region. 

For the other regions, the insiders are manageable, and the optimal generator can effectively reduce the cyber risk of the organization. 
The increase depends on the percentage of selfish insiders and honeypots.  
When the percentage of honeypots is $t^g(0)$ and insiders are all selfish, the organization's revenue with the optimal generator is $114$ times higher than the one without the optimal generator. 
Averaged over the entire region of $q^g\in [0,1]$ and $p_D^{0,H}\in [0,1]$, the organization's revenue with the optimal generator is $35.6\%$ higher than the one without the optimal generator.  
The results in  Fig. \ref{fig:combine} demonstrate that the optimal generator design provides a constructive way to quantify the \textcolor{black}{accuracy of the} information that the SOC should reveal to the insiders to establish trust with them, while in the meantime, retain her information advantage to elicit desirable insider behaviors and maximize the organization's well-being. 
These results provide a guideline to address the challenges identified in 2c and 2d of Table 2 in \cite{moore2015effective}. 

\subsubsection{Security Posture under Various Modulators}
\label{sec:investigateC}
In Section \ref{sec:investigateC}, we investigate how the (dis)incentives affect the insiders' behaviors and the security posture of the insider network. 
In Fig. \ref{fig:decisonthreshold}, we plot the decision thresholds of selfish and adversarial insiders in blue and red, respectively. 
Since the blue line has a steeper slope than the red line, Fig. \ref{fig:decisonthreshold} demonstrates that the same authentication cost affects the selfish insiders more significantly than the adversarial ones. 
As defined in  Definition \ref{def:separable}, two types of insiders are identifiable under posterior belief $p_U^H$ if $p_U^H\in [t^b(\phi^0),t^g(\phi^0)]$. 
Furthermore, a larger difference in the two thresholds, i.e., $t^g(\phi^0)-t^b(\phi^0)$, indicates a higher incentive misalignment between selfish and adversarial insiders. 
\begin{figure}[h]
\centering
\includegraphics[width=.7\columnwidth]{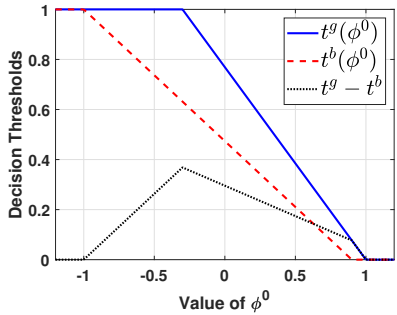} 
\caption{ 
The adversarial and the selfish  insiders' decision thresholds $t^b(\phi^0)$ and $t^g(\phi^0)$
\textcolor{black}{in the red dashed line and the blue solid line}, respectively. 
The difference $t^g(\phi^0)-t^b(\phi^0)$ denoted in the \textcolor{black}{black dotted line} represents their utility misalignment. 
}
\label{fig:decisonthreshold}
\end{figure}

Fig. \ref{fig:priorutility1} illustrates the organization's original payoff $\tilde{v}_D$ without a generator. 
The selfish insider and the SOC achieve a win-win situation at the region $\phi^0\in [0.5,0.74]$ as they both achieve their maximum payoffs at that region. 
The adversarial insider and the SOC cannot achieve a win-win situation for all $\phi^0\in \mathbb{R}$ as adversarial insiders seeking to compromise sensitive data and sabotage the organization have a completely misaligned payoff structure. 
Fig. \ref{fig:posteriorutiliy1} illustrates the organization's improved payoff $V_D$ when the optimal generator is applied. 
The results show that the optimal generator can always increase the payoffs of the selfish insiders and the organization regardless of the (dis)incentives represented by  $\phi^0\in \mathbb{R}$. 
Win-win situations still exist (resp. do not exist) for the SOC and the selfish (resp.  adversarial) insider. 

\begin{figure}[h]
    \centering 
\begin{subfigure}{0.5\columnwidth}
  \includegraphics[width=\columnwidth]{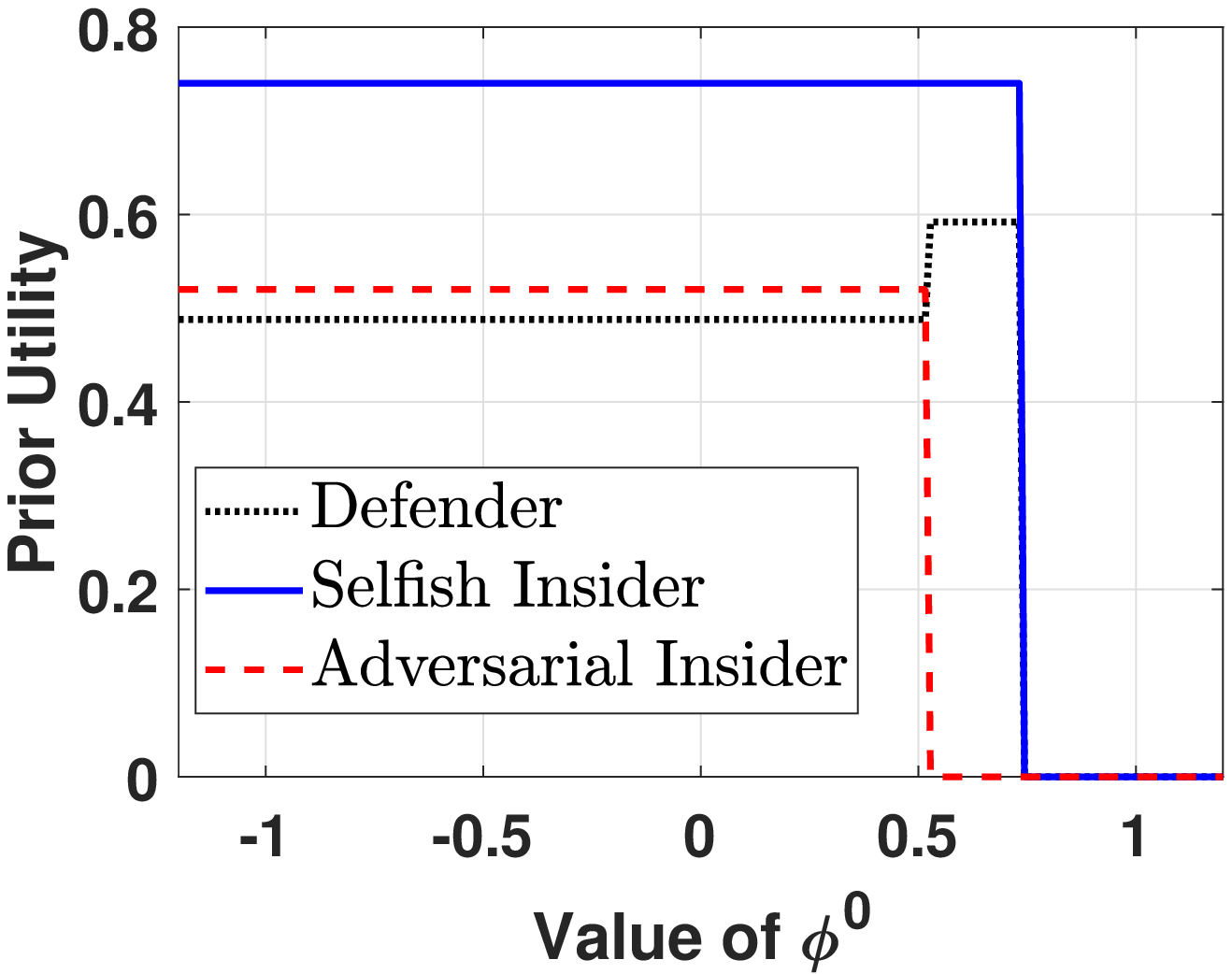}
  \caption{\label{fig:priorutility1} 
   Players' prior utilities. 
  }
\end{subfigure}\hfil 
\begin{subfigure}{0.5\columnwidth}
  \includegraphics[width=\columnwidth]{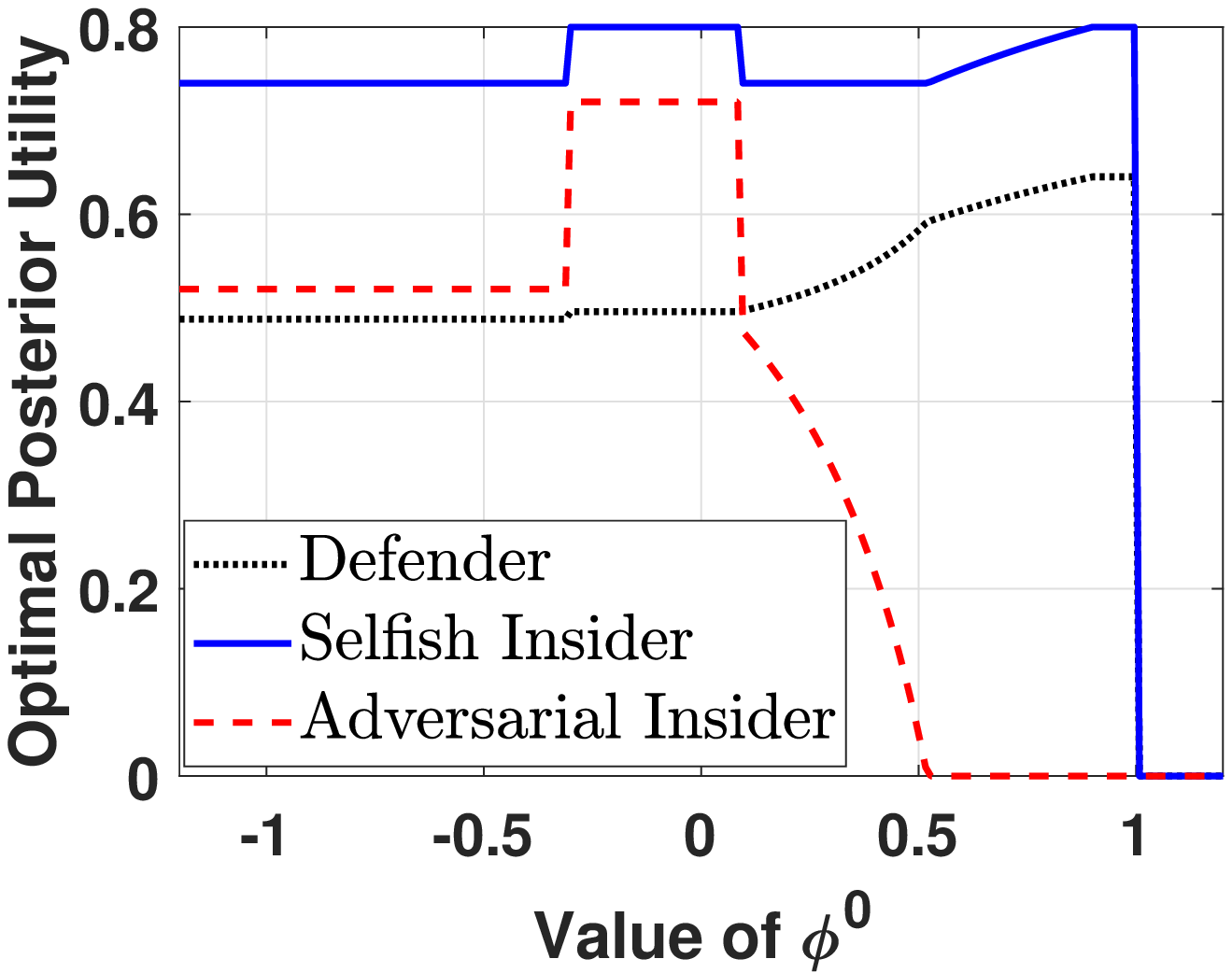}
  \caption{\label{fig:posteriorutiliy1} 
 Optimal posterior utilities. 
 }
\end{subfigure}\hfil 
\caption{
Utilities of the SOC, selfish insiders, and adversarial insiders in \textcolor{black}{the dotted
black, the solid blue, and the dashed red lines}, respectively. 
\label{fig:winwin}
}
\end{figure}

\subsubsection{Security Posture under the Covert and Overt Trust Manipulators}
\label{sec:covertandovert}


In Section \ref{sec:covertandovert}, the SOC can generate ambiguous or fake reports of the honeypot percentage so that the insiders' initial beliefs of the honeypot percentage deviate from the truth, i.e., $p_U^{0,H}\neq p_D^{0,H}$. 
Figs. \ref{fig:Exantemanipulator} and \ref{fig:Expostmanipulator} illustrate the SOC's payoffs with and without the optimal generator, respectively, under different values of $p_U^{0,H}$ and $p_D^{0,H}$. 
In Fig. \ref{fig:Exantemanipulator}, the insiders' initial beliefs fall into the following three regions. 
If $p_U^{0,H}\in [t^g(0),1]$, both types of insiders choose not to access the node. Then, the SOC's normalized payoff $\tilde{v}_D$ is zero regardless of the true percentage of honeypots $p_D^{0,H}$.  
If $p_U^{0,H}\in [t^b(0),t^g(0)]$, selfish insiders choose $a_{AC}$ and adversarial insiders choose $a_{DO}$. Then, reducing the percentage of honeypots increases the SOC's normalized payoff $\tilde{v}_D$ as it reduces the false alarm rate when selfish insiders access the honeypots. 
If $p_U^{0,H}\in [0,t^b(0)]$, both types of insiders choose to access the node. Then, reducing the percentage of honeypots also increases the SOC's normalized payoff $\tilde{v}_D$. However, the increase rate is lower than the one in the second region as the two types of insiders take the same action and are not identifiable. 

\begin{figure}[h]
    \centering 
    \begin{subfigure}{0.35\textwidth}
\includegraphics[width=1\columnwidth,height=.8\columnwidth]{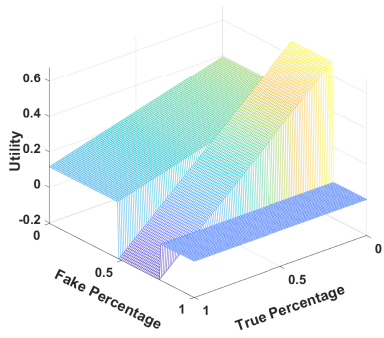}
\caption{ 
Prior utility $\tilde{v}_D$. 
}
\label{fig:Exantemanipulator}
\end{subfigure}\hfil 
    \begin{subfigure}{0.35\textwidth}
\includegraphics[width=1\columnwidth,height=.8\columnwidth]{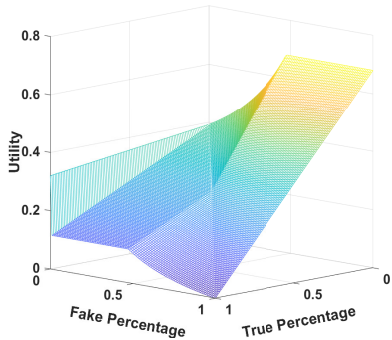} 
\caption{ 
Optimal posterior utility.  
}
\label{fig:Expostmanipulator}
\end{subfigure}\hfil 
\caption{
SOC's utilities vs. $p_D^{0,H}\in [0,1]$ and $p_U^{0,H}\in [0,1]$. 
\label{fig:trustmanipulator}
}
\end{figure}

These results illustrate that without a deceptive generator, the SOC may not always benefit from faking the percentage of honeypots. 
On the contrary, when the optimal generator is applied in Fig. \ref{fig:Expostmanipulator}, the SOC can benefit from a fake percentage of honeypots for all $p_D^{0,H},p_U^{0,H}\in [0,1]$. 
Moreover, the benefit of faking honeypot percentage is a non-decreasing function of $|p_D^{0,H}-p_U^{0,H}|$. 
Thus, the SOC obtains a higher payoff $V_D$ with the optimal generator when there is a larger mismatch between the true and the fake percentages of honeypots. 
The maximum value of $V_D$ is achieved when the true percentage of honeypots is zero and the SOC makes the insiders believe that the percentage of honeypots exceeds $t^b(0)$. 
Averaged over the true percentage $p_D^{0,H}\in [0,1]$ and the fake one $p_U^{0,H}\in [0,1]$, the SOC's payoff with the optimal generator, i.e., $V_D$ is $59.3\%$ higher than her original payoff $\tilde{v}_D$.

\section{Conclusion}
\label{sec:conclusion}
In this work, we have presented a class of duplicity games (DG) to design defensive deception mechanisms for proactive network security. 
The deception mechanism is referred to as the GMM mechanism as it consists of the following three modular design components. 
The generator provides users an appropriate amount of information to procure different types of users to take actions that are favorable to the defender. 
The incentive modulator modifies the users' utilities 
to \textcolor{black}{make their incentives better aligned with the defender's.} 
The trust manipulator makes use of users' trust to impart to them the initial beliefs that can lead to desirable security outcomes. 

We have formulated and analyzed the DG using mathematical programming and graphical approaches. 
It has been shown that the defender requires at most $N$ enforceable security policies from the entire $K^M$ ones to achieve the optimal security posture, which illustrates the efficiency of the GMM mechanism. 
\textcolor{black}{
We have proposed the concept of {\it trust margin} to measure how difficult it is for a defender to elicit the desired behavioral outcome. A user is unmanageable when the maximum trust margin is zero, as no deceptive mechanisms can affect the user's behaviors.}  
We have identified a \textit{separation principle} for the modulator design and an \textit{equivalence principle} that turns the joint design of the generator and manipulator into one single design of the manipulator. 
We have applied the DG to a case study where the defender dynamically configures the honeypot to mitigate insider threats in a corporate network. 
The numerical results have shown that the GMM mechanism manages to elicit desirable actions from both selfish and adversarial insiders and reduce the cyber risk of the organization. 
In particular, the optimal generator itself can increase the defender's payoff by $35.6\%$ on average. 
Equipped with the trust manipulator that fakes the honeypot percentage, the optimal generator can further increase the defender's payoff by $59.3\%$ on average. 
\ifCLASSOPTIONcaptionsoff
  \newpage
\fi



%


\bibliographystyle{IEEEtran}
\bibliography{honeypotDesign}

\begin{thebibliography}{10}
\providecommand{\url}[1]{#1}
\csname url@samestyle\endcsname
\providecommand{\newblock}{\relax}
\providecommand{\bibinfo}[2]{#2}
\providecommand{\BIBentrySTDinterwordspacing}{\spaceskip=0pt\relax}
\providecommand{\BIBentryALTinterwordstretchfactor}{4}
\providecommand{\BIBentryALTinterwordspacing}{\spaceskip=\fontdimen2\font plus
\BIBentryALTinterwordstretchfactor\fontdimen3\font minus
  \fontdimen4\font\relax}
\providecommand{\BIBforeignlanguage}[2]{{%
\expandafter\ifx\csname l@#1\endcsname\relax
\typeout{** WARNING: IEEEtran.bst: No hyphenation pattern has been}%
\typeout{** loaded for the language `#1'. Using the pattern for}%
\typeout{** the default language instead.}%
\else
\language=\csname l@#1\endcsname
\fi
#2}}
\providecommand{\BIBdecl}{\relax}
\BIBdecl

\bibitem{jajodia2011moving}
S.~Jajodia, A.~K. Ghosh, V.~Swarup, C.~Wang, and X.~S. Wang, \emph{Moving
  target defense: creating asymmetric uncertainty for cyber threats}.\hskip 1em
  plus 0.5em minus 0.4em\relax Springer Science \& Business Media, 2011,
  vol.~54.

\bibitem{bringer2012survey}
M.~Bringer, C.~Chelmecki, and H.~Fujinoki, ``A survey: Recent advances and
  future trends in honeypot research,'' \emph{International Journal of Computer
  Network and Information Security}, vol.~4, no.~10, p.~63, 2012.

\bibitem{al2019autonomous}
E.~Al-Shaer, J.~Wei, W.~Kevin, and C.~Wang, \emph{Autonomous Cyber
  Deception}.\hskip 1em plus 0.5em minus 0.4em\relax Springer, 2019.

\bibitem{instance1290}
``Cyber deception significantly reduces data breach costs \& improves soc
  efficiency,'' DECEPTIVE DEFENSE, INC., Tech. Rep., 08 2020.

\bibitem{Mitigation}
S.~Harris, ``Insider threat mitigation guide,'' Cybersecurity and
  Infrastructure Security Agency, Tech. Rep.

\bibitem{spitzner2003honeypots}
L.~Spitzner, ``Honeypots: Catching the insider threat,'' in \emph{19th Annual
  Computer Security Applications Conference, 2003. Proceedings.}\hskip 1em plus
  0.5em minus 0.4em\relax IEEE, 2003, pp. 170--179.

\bibitem{aumann1995repeated}
R.~J. Aumann, M.~Maschler, and R.~E. Stearns, \emph{Repeated games with
  incomplete information}.\hskip 1em plus 0.5em minus 0.4em\relax MIT press,
  1995.

\bibitem{kamenica2011bayesian}
E.~Kamenica and M.~Gentzkow, ``Bayesian persuasion,'' \emph{American Economic
  Review}, vol. 101, no.~6, pp. 2590--2615, 2011.

\bibitem{zhao2020finite}
Y.~Zhao, L.~Huang, C.~Smidts, and Q.~Zhu, ``Finite-horizon semi-markov game for
  time-sensitive attack response and probabilistic risk assessment in nuclear
  power plants,'' \emph{Reliability Engineering \& System Safety}, p. 106878,
  2020.

\bibitem{manshaei2013game}
M.~H. Manshaei, Q.~Zhu, T.~Alpcan, T.~Bac{\c{s}}ar, and J.-P. Hubaux, ``Game
  theory meets network security and privacy,'' \emph{ACM Computing Surveys
  (CSUR)}, vol.~45, no.~3, pp. 1--39, 2013.

\bibitem{huang2017large}
L.~Huang, J.~Chen, and Q.~Zhu, ``A large-scale markov game approach to dynamic
  protection of interdependent infrastructure networks,'' in
  \emph{International Conference on Decision and Game Theory for
  Security}.\hskip 1em plus 0.5em minus 0.4em\relax Springer, 2017, pp.
  357--376.

\bibitem{pawlick2021game}
J.~Pawlick and Q.~Zhu, \emph{Game Theory for Cyber Deception: From Theory to
  Applications}.\hskip 1em plus 0.5em minus 0.4em\relax Springer Nature, 2021.

\bibitem{pawlick2018modeling}
J.~Pawlick, E.~Colbert, and Q.~Zhu, ``Modeling and analysis of leaky deception
  using signaling games with evidence,'' \emph{IEEE Transactions on Information
  Forensics and Security}, vol.~14, no.~7, pp. 1871--1886, 2018.

\bibitem{9447822}
H.~Sasahara and H.~Sandberg, ``Epistemic signaling games for cyber deception
  with asymmetric recognition,'' \emph{IEEE Control Systems Letters}, vol.~6,
  pp. 854--859, 2022.

\bibitem{HUANG2020101660}
L.~Huang and Q.~Zhu, ``A dynamic games approach to proactive defense strategies
  against advanced persistent threats in cyber-physical systems,''
  \emph{Comput. \& Secur.}, vol.~89, p. 101660, 2020.

\bibitem{huang9494340}
------, ``A dynamic game framework for rational and persistent robot deception
  with an application to deceptive pursuit-evasion,'' \emph{IEEE Transactions
  on Automation Science and Engineering}, pp. 1--15, 2021.

\bibitem{feng2017signaling}
X.~Feng, Z.~Zheng, D.~Cansever, A.~Swami, and P.~Mohapatra, ``A signaling game
  model for moving target defense,'' in \emph{IEEE conference on computer
  communications}.\hskip 1em plus 0.5em minus 0.4em\relax IEEE, 2017, pp. 1--9.

\bibitem{cranford2018learning}
E.~Cranford, C.~Lebiere, C.~Gonzalez, S.~Cooney, P.~Vayanos, and M.~Tambe,
  ``Learning about cyber deception through simulations: Predictions of human
  decision making with deceptive signals in stackelberg security games.'' in
  \emph{CogSci}, 2018.

\bibitem{xu2016signaling}
H.~Xu, R.~Freeman, V.~Conitzer, S.~Dughmi, and M.~Tambe, ``Signaling in
  bayesian stackelberg games.'' in \emph{AAMAS}, 2016, pp. 150--158.

\bibitem{horak2017manipulating}
K.~Hor{\'a}k, Q.~Zhu, and B.~Bo{\v{s}}ansk{\`y}, ``Manipulating adversary’s
  belief: A dynamic game approach to deception by design for proactive network
  security,'' in \emph{GameSec}, 2017, pp. 273--294.

\bibitem{7539363}
P.~Naghizadeh and M.~Liu, ``Opting out of incentive mechanisms: A study of
  security as a non-excludable public good,'' \emph{IEEE Transactions on
  Information Forensics and Security}, vol.~11, no.~12, pp. 2790--2803, 2016.

\bibitem{7265043}
Y.~Zhang, H.~Zhang, S.~Tang, and S.~Zhong, ``Designing secure and dependable
  mobile sensing mechanisms with revenue guarantees,'' \emph{IEEE Transactions
  on Information Forensics and Security}, vol.~11, no.~1, pp. 100--113, 2016.

\bibitem{8355788}
J.~Lu, Y.~Xin, Z.~Zhang, X.~Liu, and K.~Li, ``Game-theoretic design of optimal
  two-sided rating protocols for service exchange dilemma in crowdsourcing,''
  \emph{IEEE Transactions on Information Forensics and Security}, vol.~13,
  no.~11, pp. 2801--2815, 2018.

\bibitem{7180387}
C.~Jiang, Y.~Chen, Q.~Wang, and K.~R. Liu, ``Data-driven auction mechanism
  design in iaas cloud computing,'' \emph{IEEE Transactions on Services
  Computing}, vol.~11, no.~5, pp. 743--756, 2018.

\bibitem{8355763}
Z.~Zhang, S.~He, J.~Chen, and J.~Zhang, ``Reap: An efficient incentive
  mechanism for reconciling aggregation accuracy and individual privacy in
  crowdsensing,'' \emph{IEEE Transactions on Information Forensics and
  Security}, vol.~13, no.~12, pp. 2995--3007, 2018.

\bibitem{8913631}
R.~Zhang and Q.~Zhu, ``$\mathtt{FlipIn}$ : A game-theoretic cyber insurance
  framework for incentive-compatible cyber risk management of internet of
  things,'' \emph{IEEE Transactions on Information Forensics and Security},
  vol.~15, pp. 2026--2041, 2020.

\bibitem{chen2017security}
J.~Chen and Q.~Zhu, ``Security as a service for cloud-enabled internet of
  controlled things under advanced persistent threats: a contract design
  approach,'' \emph{IEEE Transactions on Information Forensics and Security},
  vol.~12, no.~11, pp. 2736--2750, 2017.

\bibitem{rabinovich2015information}
Z.~Rabinovich, A.~X. Jiang, M.~Jain, and H.~Xu, ``Information disclosure as a
  means to security,'' in \emph{Proceedings of the 2015 International
  Conference on Autonomous Agents and Multiagent Systems}.\hskip 1em plus 0.5em
  minus 0.4em\relax Citeseer, 2015, pp. 645--653.

\bibitem{das2017reducing}
S.~Das, E.~Kamenica, and R.~Mirka, ``Reducing congestion through information
  design,'' in \emph{2017 55th annual allerton conference on communication,
  control, and computing (allerton)}.\hskip 1em plus 0.5em minus 0.4em\relax
  IEEE, 2017, pp. 1279--1284.

\bibitem{HORAK2019101579}
K.~Horák, B.~Bošanský, P.~Tomášek, C.~Kiekintveld, and C.~Kamhoua,
  ``Optimizing honeypot strategies against dynamic lateral movement using
  partially observable stochastic games,'' \emph{Computers \& Security},
  vol.~87, p. 101579, 2019.

\bibitem{moore2015effective}
A.~P. Moore, W.~Novak, M.~Collins, R.~Trzeciak, and M.~Theis, ``Effective
  insider threat programs: understanding and avoiding potential pitfalls,''
  \emph{Software Engineering Institute White Paper, Pittsburgh}, 2015.

\bibitem{KANTZAVELOU2010859}
I.~Kantzavelou and S.~Katsikas, ``A game-based intrusion detection mechanism to
  confront internal attackers,'' \emph{Computers \& Security}, vol.~29, no.~8,
  pp. 859--874, 2010.

\bibitem{LIU200875}
``Game-theoretic modeling and analysis of insider threats,''
  \emph{International Journal of Critical Infrastructure Protection}, vol.~1,
  pp. 75--80, 2008.

\bibitem{9218982}
C.~Joshi, J.~R. Aliaga, and D.~R. Insua, ``Insider threat modeling: An
  adversarial risk analysis approach,'' \emph{IEEE Transactions on Information
  Forensics and Security}, vol.~16, pp. 1131--1142, 2021.

\bibitem{casey2015compliance}
W.~A. Casey, Q.~Zhu, J.~A. Morales, and B.~Mishra, ``Compliance control:
  Managed vulnerability surface in social-technological systems via signaling
  games,'' in \emph{Proceedings of the 7th ACM CCS International Workshop on
  Managing Insider Security Threats}, 2015, pp. 53--62.

\bibitem{yamin2019implementation}
M.~M. Yamin, B.~Katt, K.~Sattar, and M.~B. Ahmad, ``Implementation of insider
  threat detection system using honeypot based sensors and threat analytics,''
  in \emph{Future of Information and Communication Conference}.\hskip 1em plus
  0.5em minus 0.4em\relax Springer, 2019, pp. 801--829.

\bibitem{dahbul2017enhancing}
R.~Dahbul, C.~Lim, and J.~Purnama, ``Enhancing honeypot deception capability
  through network service fingerprinting,'' in \emph{Journal of Physics:
  Conference Series}, vol. 801, no.~1.\hskip 1em plus 0.5em minus 0.4em\relax
  IOP Publishing, 2017, p. 012057.

\bibitem{morishita2019detect}
S.~Morishita, T.~Hoizumi, W.~Ueno, R.~Tanabe, C.~Ga{\~n}{\'a}n, M.~J. van
  Eeten, K.~Yoshioka, and T.~Matsumoto, ``Detect me if you… oh wait. an
  internet-wide view of self-revealing honeypots,'' in \emph{2019 IFIP/IEEE
  Symposium on Integrated Network and Service Management (IM)}.\hskip 1em plus
  0.5em minus 0.4em\relax IEEE, 2019, pp. 134--143.

\bibitem{WinNT3}
\BIBentryALTinterwordspacing
Verizon. (2017) Vocabulary for event recording and incident sharing (veris).
  [Online]. Available: \url{http://veriscommunity.net/}
\BIBentrySTDinterwordspacing

\bibitem{shi2019dynamic}
L.~Shi, Y.~Li, T.~Liu, J.~Liu, B.~Shan, and H.~Chen, ``Dynamic distributed
  honeypot based on blockchain,'' \emph{IEEE Access}, vol.~7, pp.
  72\,234--72\,246, 2019.

\bibitem{wagener2011adaptive}
G.~Wagener, R.~State, T.~Engel, and A.~Dulaunoy, ``Adaptive and
  self-configurable honeypots,'' in \emph{12th IFIP/IEEE International
  Symposium on Integrated Network Management (IM 2011) and Workshops}.\hskip
  1em plus 0.5em minus 0.4em\relax IEEE, 2011, pp. 345--352.

\bibitem{huang2019adaptive}
L.~Huang and Q.~Zhu, ``Adaptive honeypot engagement through reinforcement
  learning of semi-markov decision processes,'' in \emph{International
  Conference on Decision and Game Theory for Security}.\hskip 1em plus 0.5em
  minus 0.4em\relax Springer, 2019, pp. 196--216.

\bibitem{rowe2007defending}
N.~C. Rowe, E.~J. Custy, and B.~T. Duong, ``Defending cyberspace with fake
  honeypots.'' \emph{JCP}, vol.~2, no.~2, pp. 25--36, 2007.

\bibitem{orlik2013arrangements}
P.~Orlik and H.~Terao, \emph{Arrangements of hyperplanes}.\hskip 1em plus 0.5em
  minus 0.4em\relax Springer Science \& Business Media, 2013, vol. 300.

\end{thebibliography}

%




\begin{IEEEbiography}[{\includegraphics[width=1in,height=1.25in,clip,keepaspectratio]{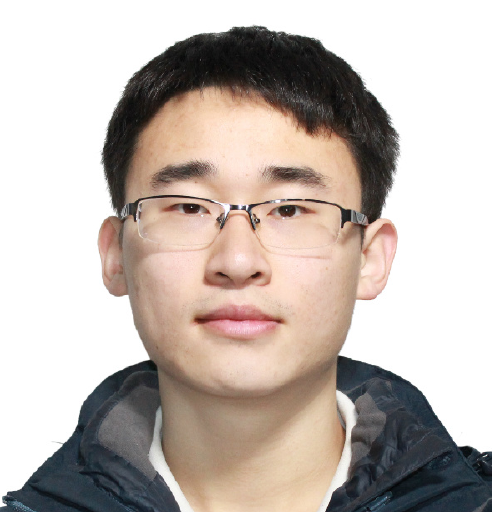}}]{Linan Huang}
(S'16) received the B.Eng. degree (Hons.) in Electrical Engineering from Beijing Institute of Technology, China, in 2016. He is currently pursuing a Ph.D. degree at the Laboratory for Agile and Resilient Complex Systems, Tandon School of Engineering, New York University, NY, USA.
His research interests include dynamic decision-making of the multi-agent system, mechanism design, artificial intelligence, security, and resilience for the cyber-physical systems. 
\end{IEEEbiography}
\begin{IEEEbiography}[{\includegraphics[width=1in,height=1.25in,clip,keepaspectratio]{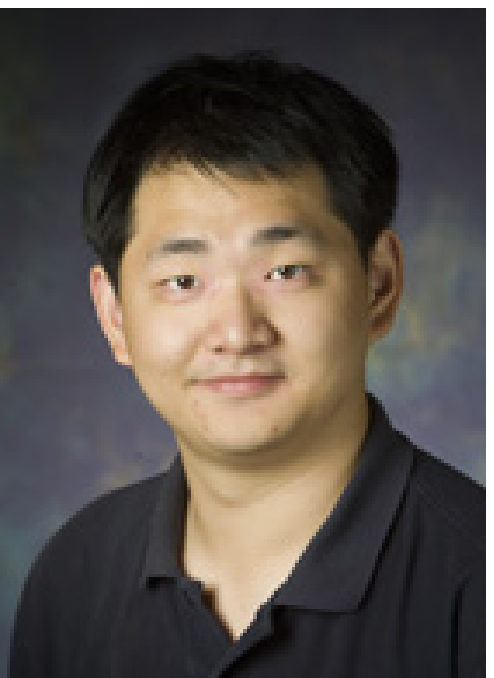}}]{Quanyan Zhu}
(SM’02-M’14) received B. Eng. in Honors Electrical Engineering from McGill University in 2006, M. A. Sc. from the University of Toronto in 2008, and Ph.D. from the University of Illinois at Urbana-Champaign (UIUC) in 2013. 
After stints at Princeton University, he is currently an associate professor at the Department of Electrical and Computer Engineering, New York University (NYU). He is an affiliated faculty member of the Center for Urban Science and Progress (CUSP) and Center for Cyber Security (CCS) at NYU. His current research interests include game theory, machine learning, cyber deception, and cyber-physical systems.
\end{IEEEbiography}

\end{document}